\documentclass[onecolumn,10pt]{IEEEtran}

\makeatletter
\usepackage{graphicx}
\usepackage[all]{xy}
\usepackage[latin1]{inputenc}
\usepackage{amsfonts,amsmath,amssymb,amsthm}
\usepackage{indentfirst,latexsym,amscd}
\usepackage{amsbsy}
\usepackage{graphicx,multirow,bm}
\usepackage{float}
\usepackage{mathrsfs}

\renewcommand{\arraystretch}{1.7}

\newtheorem{lemma}{Lemma}
\newtheorem{theorem}{Theorem}

\newtheorem{rem}{Remark}

\title{A class of asymptotically  optimal authentication codes with secrecy}
\author{Haibo Liu, Chengzhi Wei, and Qunying Liao%
\thanks{The work of H.B. Liu is supported in part by the Sichuan Natural Science Foundation with No. 2024NSFSC0417 and No. 2026NSFSC0138; the work of Q.Y. Liao is supported in part by the National Natural Science Foundation of China with No. 12471494.}%
\thanks{H.B.Liu, School of Applied Mathematics, Chengdu University of Information Technology,
 \small Chengdu, Sichuan, China (Email: liuhaibo@cuit.edu.cn).}%
\thanks{C.Z. Wei, School of Applied Mathematics, Chengdu University of Information Technology,
 \small Chengdu, Sichuan, China (Email: 2263759746@qq.com).}%
\thanks{Q.Y. Liao, School of Mathematical Sciences, Sichuan Normal University,
 \small Chengdu, Sichuan, China (Email: qunyingliao@sicnu.edu.cn).}%
}
\date{}

\begin{document}
\baselineskip15pt \maketitle

\begin{abstract}

  In this paper, a class of linear authentication codes with secrecy, which are equipped with simple encoding rules, is constructed. By means of a special Weil sum, the maximum success probabilities of impersonation attack (denoted by $P_I$) and of substitution attack (denoted by $P_S$) for these codes are explicitly derived. It is further proven that the codes are asymptotically optimal, in the sense that both the information theoretic bound and the combinatorial bound for authentication codes are asymptotically attained by $P_I$ and $P_S$.


\end{abstract}
\section{INTRODUCTION}
$ $
The authentication model established by Simmons\cite{ref1} involves three parties: a transmitter, a receiver, and an opponent. The transmitter wants to send some information namely source states to the receiver through a public communication channel. The sender and receiver share a secret key $k\in\mathscr{K}$, where $\mathscr{K}$ is the key space. Each secret key $k$ corresponds to a encoding rule $E_{k}$, the encoding rule $E_{k}$ is a one-to-one transformation, which transforms a source state $s$ into a message $m=E_{k}(s)$, and then sends $m$ to the receiver through the channel. All the possible source states $s$ form the source state space $\mathscr{S}$, and all the possible messages $m$ form a space $\mathscr{M}$, which is called the message space. All the possible encoding rules form the encoding rules space $\mathscr{E} =\left \{ E_{k}:k\in \mathscr{K}   \right \} $. The source state space and the key space are associated with a probability distribution. In this work, we assume that the probability distribution for both state space and the key space is uniform. An authentication code is defined by the four-tuple $(\mathscr{S} ,\mathscr{K} ,\mathscr{M},\mathscr{E})$.

There are two types of authentication codes: authentication codes with secrecy and those without secrecy. In an authentication code with secrecy, the secret key $k$ shared by the sender and receiver is used for both encryption and authentication purpose. In this case, the source state $s$ is hidden in the encoded message $m$, and one cannot recover $s$ from $m$ without the knowledge of $k$ or $E_k$. In an authentication code without secrecy, the source state is sent to the receiver in plaintext, and the secret key is used only for authentication purpose. In this paper, we focus on only authentication codes with secrecy. It is possible that an encoding rule may map a source state into more than one message (this is called splitting). In this paper, we consider only authentication codes without splitting. With this authentication model, we consider two kinds of attacks, the impersonation and substitution attacks. In the impersonation attack, an opponent can insert his message into the channel, and hopes the receiver will accept it as authentic. In the substitution attack, the opponent observes a message sent by the transmitter, and replaces it with his message, and wishes the receiver will accept it as authentic. We use $P_{I}$ and $P_{S}$ to denote the maximum success probabilities with respect to the two types of attacks.

Several approaches had been proposed for constructing authentication codes, both with and without secrecy. Most of these constructions relied on combinatorial designs and finite geometries (see \cite{ref5,ref6,ref7,ref8,ref9,ref10,ref11,ref12,ref13}), but they were generally difficult to implement. Ding et al. \cite{ref3} presented three algebraic constructions of authentication codes with secrecy, computed $P_I$, and established an upper bound on $P_S$. Moreover, using results from Gaussian sums and quadratic forms, they proved that $P_I$ was asymptotically optimal with respect to the information theoretic bound and the combinatorial bound. Later, Ding and Tian \cite{ref2} employed the trace function over finite fields, authentication perpendicular arrays over finite fields, and perfect nonlinear functions over finite abelian groups to construct three classes of authentication codes with perfect secrecy, and they proved that these codes were optimal. In \cite{ref14}, Ding and Niederreiter utilized perfect and almost-perfect nonlinear functions to obtain four classes of authentication codes without secrecy. Based on preimage results for these functions, they gave upper bounds for both $P_I$ and $P_S$. Subsequently, they used functions with high nonlinearity \cite{ref15} to construct two families of authentication codes and provided upper bounds on $P_I$ and $P_S$, which improved upon the constructions in \cite{ref14}. More recently, Juan and Javier \cite{ref16} derived seven authentication codes with secrecy from functions over finite fields. They computed $P_I$ and, by means of preimage results for those functions, gave an upper bound on $P_S$, thereby generalizing, in a certain sense, the constructions of \cite{ref2}.

Authentication codes constructed by combinatorial designs and finite geometries are optimal, but they are difficult to implement. Whereas algebraically constructed authentication codes are easy to implement, yet it is difficult to compute $P_I$ and $P_S$ and to prove that they are optimal or asymptotically optimal. In this paper, motivated by \cite{ref3}, we present a new algebraic construction of authentication codes with secrecy, the codes have simple algebraic structures. Based on the special Weil sum, we explicitly calculate $P_{I}$ and $P_{S}$ of these codes, and prove they are asymptotically optimal with respect to the information theoretic bound and the combinatorial bound of $P_I$ and $P_S$.

The paper is organized as follows. Section 2 presents certain bounds for authentication codes along with known results on exponential sums that are needed later. Section 3 gives a new algebraic construction of authentication codes with secrecy, explicitly calculates $P_{I}$ and $P_{S}$, and proves these codes are asymptotically optimal in the sense that both $P_I$ and $P_S$ asymptotically attain the information theoretic bound and the combinatorial bound for authentication codes. Section 4 concludes the whole paper.

\section{Preliminaries}
$ $
 In this section, we provide some bounds for authentication codes and some related known results about exponential sums, which will be needed in the sequel.

\subsection{Bounds on authentication codes}

First, we introduce some bounds on authentication codes. To this end, we use $\mathscr{S} ,\mathscr{M}$ and $\mathscr{E}$ to denote the random variables of the source states, messages and encoding rules. Let $\mathscr{M}^{r} $ to denote the random variables of the first $r$ messages, where $r$ is a positive integer, and $H\left (\mathscr{E}|\mathscr{M}^{r} \right )$ to denote the conditional entropy of $\mathscr{E}$ given that the first $r$ messages have been observed. The following lemma is called the information-theoretic bound of authentication codes.
\begin{lemma}\cite{ref2}\label{lem8}
Let the notations be defined as before, then for any authentication code, we have
\[ P_{I} \geq 2^{H(\mathscr{E}|\mathscr{M})-H(\mathscr{E} )}, \qquad P_{S} \geq 2^{H(\mathscr{E}|\mathscr{M}^{2})-H( \mathscr{E} |\mathscr{M})},\]
where the entropies are measured in bits.
\end{lemma}
The following lemma is called the combinatorial bound of authentication codes.
\begin{lemma}\cite{ref2}\label{lem9}
Let the notations be defined as before, then for any authentication code without splitting, we have
\[P_{I} \geq \frac{\left | \mathscr{S}  \right | }{\left | \mathscr{M} \right | } \quad\text{and}\quad  P_{S} \geq \frac{\left | \mathscr{S}  \right |-1 }{\left | \mathscr{M} \right |-1 }.\]
If both equalities are achieved, then $\left | \mathscr{S} \right |\geq \left |\mathscr{M} \right |$.
\end{lemma}

\subsection{Exponential sums}

In this subsection, we present some related results for exponential sums over finite fields. For more details, readers are referred to Chapter 5 in \cite{ref17}.

Let $p$ be a prime, and $q=p^n$, where $n$ is a positive integer. For each $b\in\mathbb{F}_{p^n}$, the function $\chi_b(x)=\zeta_p^{\textrm{Tr}_n(bx)}$ defines an additive character of $\mathbb{F}_{p^n}$,
where $\zeta_p=\textrm{e}^{\frac{2\pi\sqrt{-1}}{p}}$ is a primitive $p$-th root of unity and $x\in \mathbb{F}_{p^n}$. Specially, when $b=0$, $\chi_0(x)=1$ for any $x\in\mathbb{F}_{p^n}$, which is called the trivial additive character of $\mathbb{F}_{p^n}$. When $b=1$, the character $\chi_1$ is called the canonical additive character of $\mathbb{F}_{p^n}$, and every additive character of $\mathbb{F}_{p^n}$ can be written as $\chi_b(x)=\chi_1(bx)$. Let $g$ be a fixed primitive element of $\mathbb{F}_{p^n}^{\ast}$, the multiplicative character of $\mathbb{F}_{p^n}$ is defined by $\psi_{j}(g^{k})=\zeta_{p}^{\frac{jk}{p^n-1}}$($k=0,1,\ldots,p^n-2, 0\leq j\leq p^n-2$). In particular, $\psi_{\frac{p^n-1}{2}}$ is called the quadratic multiplicative character of $\mathbb{F}_{p^n}$, denoted by $\eta_{n}$. The additive character and the quadratic multiplicative character over $\mathbb{F}_{p^n}$ have the following properties,
$$\sum_{x \in \mathbb{F}_{p^n}}\zeta_{p}^{\textrm{Tr}_n(bx)}=\left\{\begin{array}{ll} p^n, &\quad b=0,\\
0, &\quad \text{otherwise};
\end{array}\right.$$
and for any $x \in \mathbb{F}_{p}^{*}$,
$$\eta_{n}(x)=\left\{\begin{array}{cc} 1, &\quad 2|n, \\
\eta_{1}(x), &\quad \text{otherwise};
\end{array}\right.$$
where $\eta_{1}(x)$ is the quadratic multiplicative character of $\mathbb{F}_p$. The Gauss sum $G_{n}$ over $\mathbb{F}_{p^n}$ is defined by $G_{n}=\sum_{x\in\mathbb{F}_{p^n}}\eta_{n}(x)\chi_{1}(x)$, for the value of $G_{n}$, we have the following lemma.
\begin{lemma}\label{lem14}(\cite{ref17})
Let $p$ be an odd prime, and $n$ be a positive integer. Then
\[G_{n}=(-1)^{n-1}L^{n}p^{\frac{n}{2}},\]
where $L=(-1)^{\frac{(p-1)^2}{8}}$.
\end{lemma}
For $f(x) \in \mathbb{F}_{p^{n}}[x]$, the Weil sum with $f(x)$ over $\mathbb{F}_{p^{n}}$ is defined by $\sum_{x \in \mathbb{F}_{p^{n}}}{\chi(f(x))}$. In particular, for a positive integer $u$, we define the special Weil sum by
\[S_{n,u}(a,b)=\sum_{x \in \mathbb{F}_{p^{n}}}{\chi(ax^{p^u+1}+bx)}\quad (a\in \mathbb{F}_{p^{n}}^\ast, b\in  \mathbb{F}_{p^{n}}).\]
When $b=0$, the value of $S_{n,u}(a,0)$ is determined by the following lemma.
\begin{lemma}(\cite{ref4})\label{lem2}
Let $u$ and $n$ be two positive integers, $a\in \mathbb{F}_{p^{n}}^{*}$, and denote $gcd(n,u)=v$. If $\frac{n}{v}\equiv1 \,(mod\,2)$ , then \[S_{n_,u}(a,0)=G_{n}\eta_{n}(a);\]
if $\frac{n}{v}\equiv0 \,(mod\,2)$ , then
$$S_{n,u}(a,0)=\left\{\begin{array}{ll}p^s,
   & \text{if $\frac{n}{v} \equiv0\,(mod\,4)$ and} \, a^{\frac{p^{n}-1}{p^v+1}}\neq(-1)^{\frac{n}{2v}},\\
   -p^{s+v}, & \text{if $\frac{n}{v} \equiv0\,(mod\,4)$ and} \, a^{\frac{p^{n}-1}{p^v+1}}=(-1)^{\frac{n}{2v}}, \\
   -p^s, &\text{if $\frac{n}{v} \equiv2\,(mod\,4)$ and}\, a^{\frac{p^{n}-1}{p^v+1}}\neq(-1)^{\frac{n}{2v}},\\
   p^{s+v}, &\text{if $\frac{n}{v} \equiv2\,(mod\,4)$ and}\, a^{\frac{p^{n}-1}{p^v+1}}=(-1)^{\frac{n}{2v}}.
   \end{array}\right.$$
\end{lemma}
When $b\neq0$, the value of $S_{n,u}(a,b)$ depends on the existence of the solution to the equation
\begin{equation}\label{equal3}
a^{p^u}X^{p^{2u}}+aX=-b^{p^{u}}
\end{equation}
\begin{lemma}(\cite{ref4})\label{lem1}
Let $n,u$ and $v$ be defined as in Lemma \ref{lem2}. for $a\in \mathbb{F}_{p^{n}}^\ast$, $b\in  \mathbb{F}_{p^{n}}$, if $f(X)=a^{p^u}X^{p^{2u}}+aX$ is a permutation polynomial over $\mathbb{F}_{p^{n}}$, then
   $$S_{n,u}(a,b)=\left\{\begin{array}{ll}G_{n}\eta_{n}(a)\zeta_{p}^{\textrm{Tr}_{n}(-ax_{0}^{p^u+1})},
   & \text{if $\frac{n}{v}\equiv1 \,(mod\,2)$ },\\
   (-1)^{\frac{n}{2v}}p^{\frac{n}{2}}\zeta_p^{\textrm{Tr}_{n}(-ax_{0}^{p^u+1})}, & \,\text{if $\frac{n}{v}\equiv0 \,(mod\,2)$ },
   \end{array}\right.$$
where $x_{0}\in \mathbb{F}_{p^{n}}$ is the unique solution of $f(X)=-b^{p^u}$.
\end{lemma}
\begin{lemma}\label{lem4}
(\cite{ref4})
For $a\in \mathbb{F}_{p^{n}}$, the equation $$a^{p^u}X^{p^{2u}}+aX=0$$ is solvable in $\mathbb{F}_{p^{n}}^*$ if and only if both $\frac{n}{v}\equiv0 \,(mod\,2)$ and $a^{\frac{p^{n}-1}{p^v+1}}=(-1)^{\frac{n}{2v}}$. When the condition is satisfied, the equation has $p^{2v}-1$ non-zero solutions.
\end{lemma}

\begin{rem}\label{remark1}
By lemma \ref{lem4}, it is easy to see that $f(X)=a^{p^u}X^{p^{2u}}+aX$ is a permutation polynomial over $\mathbb{F}_{p^{n}}$ if and only if $\frac{n}{v}\equiv1 \,(mod\,2)$ or both $\frac{n}{v} \equiv0\,(mod\,2)$ and $a^{\frac{p^{n}-1}{p^v+1}}\neq(-1)^{\frac{n}{2v}}$.
\end{rem}

When $f(X)=a^{p^u}X^{p^{2u}}+aX$ is not a permutation polynomial over $\mathbb{F}_{p^{n}}$, the value of $S_{n,u}(a,b)$ is determined by the following lemma.
\begin{lemma}(\cite{ref4})\label{lem5}
Let $n,u$ and $v$ be defined as in Lemma \ref{lem2}, for $a\in \mathbb{F}_{p^{n}}^\ast$, $b\in \mathbb{F}_{p^{n}}$, if $f(X)=a^{p^u}X^{p^{2u}}+aX$ is not a permutation polynomial over $\mathbb{F}_{p^{n}}$, then
$$S_{n,u}(a,b)=\left\{\begin{array}{ll}
   -(-1)^{\frac{n}{2v}}p^{\frac{n}{2}+v}\zeta_p^{\textrm{Tr}_{n}(-ax_{0}^{p^u+1})}, & \text{if}\, f(X)=-b^{p^u} \,\text{is solvable},\\
  0,& \text{otherwise},
   \end{array}\right.$$
where $x_{0}$ is the solution to the equation $f(X)=-b^{p^u}$.
\end{lemma}
Particularly, the number of $b$ such that the equation $f(X)=-b^{p^u}$ is solvable over $\mathbb{F}_{p^{n}}$ is determined by the following lemma.
\begin{lemma}\label{lem6}
(\cite{ref4})
Let $n,u$ and $v$ be defined as in Lemma \ref{lem2}, for $a\in \mathbb{F}_{p^{n}}^\ast$, denote $B=\{b\in \mathbb{F}_{p^{n}} | \,a^{p^u}X^{p^{2u}}+aX=-b^{p^u}\ \text{is solvable in}\ \mathbb{F}_{p^{n}}\}$, if $a^{\frac{p^{n}-1}{p^v+1}}=(-1)^{\frac{n}{2v}}$, then $\#B=p^{n-2v}$.
\end{lemma}
The character sum of any quadratic polynomial is given by the following lemma.
\begin{lemma}(\cite{ref17}, Theorem 5.33)\label{lem12}
Let $q=p^n$ be odd and $f(x)=a_{2}x^2+a_{1}x+a_{0}\in \mathbb{F}_{q}[x]$ with $a_{2}\neq 0$. Then\\
$$\sum_{x\in \mathbb{F}_{q}}{\chi(f(x))}=\chi(a_{0}-a_{1}^2(4a_{2})^{-1})\eta_n(a_{2})G_n.$$
\end{lemma}
\section{Main research result}

For two positive integers $n(n\geq 4)$ and $r$, let $Tr(x)$ be the absolute trace function from $ \mathbb{F} _{p^{n} } $ to $\mathbb{F} _{p}$, where $p$ is an odd prime,  we use $\mathscr{S}$, $\mathscr{K}$, $\mathscr{M}$, and $\mathscr{E}$ to denote the source state space, key space, message space, and encoding rule space, respectively. Define the authentication code with secrecy as
\[
    \left ( \mathscr{S} ,\mathscr{K},\mathscr{M},\mathscr{E}  \right ) =\left ( \mathbb{F} _{p^{n} } ,\mathbb{F} _{p^{n} }, \mathbb{F} _{p^{n} }\times \mathbb{F} _{p},\left \{ E_{k}|k\in \mathscr{K}  \right \}   \right ),
\]
where
\begin{align}\label{equal5}
    E_{k} \left ( s \right ) =\left ( s+  k^{p^{r} },Tr\left ( sk \right )   \right )
\end{align}
for any $k\in \mathscr{K} $ and $s\in \mathscr{S}$.  Let $m_{1} = s+  k^{p^{r} }$ and $m_{2} =Tr\left ( sk \right )$, the first part $m_1$ is the encrypted message, and the second part $m_2$ is the redundant part for authentication. Next we will calculate $P_I$ and $P_S$ of the authentication code in (\ref{equal5}).
\begin{rem}
In \cite{ref2} and \cite{ref3}, the encryption part of the authentication code is $s+k$, i.e., the encryption key is directly added to the secret. In contrast, the encryption part of this paper is $s+k^{p^r}$, the secret is  added to the image of the encryption key $k$ under a nonlinear automorphism of $\mathbb{F}_p^n$, which is therefore more secure and reliable.
\end{rem}

\subsection{Impersonation attack}

Assuming that attacker knows the overall system structure, but is unaware of the secret key $k$ or the corresponding encoding rule $E_{k}$. This section will discuss the security of the authentication code in (\ref{equal5}) against impersonation attack. In such an attack, the adversary randomly or selectively chooses an authentication code $m=\left ( m_{1},m_{2}   \right ) \in\mathscr{M}$ and sends it to the receiver. Upon receipt, the receiver computes $s=m_{1} -k^{p^{r}  } $ and $Tr\left ( sk \right ) $, and then verifies the validity of the equation $Tr\left ( sk \right ) =m_{2} $. Thus, we have
\begin{align}\label{euqal6}
    \begin{split}
    P_{I} &=\max_{m_{1}\in\mathbb{F}_{p^n},m_{2}\in\mathbb{F}_p} \frac{\left | \left \{  k\in \mathbb{F} _{p^{n} } : Tr\left ( m_{1}k- k^{p^{r}+1 } \right ) =m_{2}   \right \}  \right | }{p^{n} } \\
   &=\max_{a\in \mathbb{F} _{p^{n} } ,b\in \mathbb{F} _{p} } \frac{\left | \left \{  x\in \mathbb{F} _{p^{n} } : Tr\left ( ax- x^{p^{r}+1 } \right )=b     \right \}  \right | }{p^{n} } .
\end{split}
\end{align}

For convenience, we use $\zeta_p $ to denote a complex $p$-th root of unity, use $\eta $ and $\eta ^{\prime}$ to denote the quadratic character over $\mathbb{F} _{p^{n} } $ and the quadratic character over $\mathbb{F} _{p}$, respectively. We denote $\chi_{1} $ and $\chi_{1} ^{\prime} $ be the canonical additive character over $\mathbb{F} _{p^{n} }$ and the additive canonical character over $\mathbb{F} _{p} $, respectively. In order to calculate $P_I$, we need the following lemma.

\begin{lemma}\label{lem10}
Let $n$ and $r$ be two positive integers and $v=gcd(n,r)$, for $ a \in \mathbb{F}_{p^{n}}, b \in \mathbb{F}_{p}$, denote
\[ N(a,b)=|\left \{ x\in \mathbb{F} _{p^{n} }:Tr(ax-x^{p^{r}+  1 } )=b\right \}|,\]
then we have the following result.

(1) If $n\equiv1\,(mod\,2)$, then
    \[
   N(a,b)=\begin{cases}
   p^{n-1},& \text{ if } a= 0,b=0\ or\ a\ne  0,\lambda_a = b,\\
  p^{n-1}+ p^{-1}G_{n}G_1 \eta^{\prime } \left ( b \right ),& \text{ if } a= 0,b\ne 0,\\
  p^{n-1} + p^{-1} G_{n}G_1\eta ^{\prime }(b-\lambda_a ), & \text{ if }a\ne  0,\lambda_a \ne b;
\end{cases}
    \]

(2) if $n\equiv0\,(mod\,2)$ and $\frac{n}{v}\equiv1 \,(mod\,2)$, then
    \[
   N(a,b)=\begin{cases}
  p^{n-1}+ p^{-1}\left ( p-1 \right ) G_{n},  & \text{ if } a= 0,b=0 \ or \ a\ne  0,\lambda_a = b, \\
  p^{n-1}-p^{-1}G_{n},     & \text{ if }a= 0, b\ne0\ or\ a\ne  0,\lambda_a \ne b;  \\
\end{cases}
    \]

(3) if $\frac{n}{v} \equiv2\,(mod\,4)$, then
 \[
    N(a,b)=\begin{cases}
    p^{n-1}-(p-1) p^{\frac{n}{2}-1 },   & \text{ if } a=0,b= 0 \ or \ a\ne0,\lambda_a= b,\\
    p^{n-1}+ p^{\frac{n}{2}-1 },   & \text{ if } a=0,b\ne 0 \ or \ a\ne0,\lambda_a\ne b;
\end{cases}
\]

(4) if $\frac{n}{v} \equiv0\,(mod\,4)$, then
\[
   N(a,b)=\begin{cases}
   p^{n-1}-(p-1)p^{\frac{n}{2}+  v-1},  & \text{ if }a=0,b= 0 \ or \ a\ne 0,\lambda_a = b, (\ref{equal1})\ is\ solvable, \\
   p^{n-1} +p^{\frac{n}{2} +  v-1},     & \text{ if }a=0,b\ne 0 \ or \ a\ne 0,\lambda_a \ne b, (\ref{equal1})\ is\ solvable,  \\
   p^{n-1}, & \text{ if } a\ne 0,(\ref{equal1})\ is\ unsolvable,
\end{cases}
\]
where $\gamma_a$ is a solution of $X^{p^{2r} }+ X=- a^{p^{r} }$ and $\lambda_a=Tr\left ( \gamma_a^{p^{r}+  1 }  \right )$.
\end{lemma}

\begin{proof}
By the orthogonal property of additive character, we have
\[
    \begin{split}
        N\left ( a,b \right ) &=p^{-1}\sum_{x\in \mathbb{F} _{p^{n} }   }\sum_{y\in \mathbb{F} _{p}  }\zeta _{p}^{ y\left ( Tr\left (ax -x^{p^{r}+  1 }  \right ) -b \right )    } \\
        &=p^{n-1}+  p^{-1} \sum_{x\in \mathbb{F} _{p^{n} } }\sum_{y\in \mathbb{F} _{p}^{*} }\zeta _{p} ^{ y\left ( Tr\left ( ax-x^{p^{r}+  1 }  \right ) -b \right )    } \\
        &=p^{n-1}+  p^{-1} \sum_{x\in \mathbb{F} _{p^{n} } }\zeta _{p} ^{\left (ayx -yx^{p^{r}+  1 }  \right )  }\sum_{y\in \mathbb{F} _{p}^{*} } \zeta_p^{ -yb }  \\
        &=p^{n-1}+ p^{-1} \sum_{y\in \mathbb{F} _{p}^{*} }\zeta_p^{-by}S_{n,r} \left ( -y,ay \right )
    \end{split}
\]
When $a=0$, by Lemma \ref{lem2}, if $\frac{n}{v}\equiv1 \,(mod\,2)$ , then $S_{n,r} \left ( -y,0 \right ) =G_{n}  \eta  \left ( -y \right ) $. So,
\[
    \begin{split}
      N(a,b) &=p^{n-1} + p^{-1} \sum_{y\in \mathbb{F} _{p}^{*} }\zeta_p^{ -b y }G_{n}  \eta  \left ( -y \right )\\
       &=\begin{cases}
 p^{n-1}+p^{-1}\sum_{y\in \mathbb{F} _{p}^{*} }G_{n} \eta  \left ( -y \right ),    & \text{ if } b=0,\\
 p^{n-1}+p^{-1}\sum_{y\in \mathbb{F} _{p}^{*} }\zeta_p^{-by}G_{n}\eta (-y),       & \text{ if } b\ne 0,
\end{cases}\\
&=\begin{cases}
  p^{n-1}+ p^{-1}\left ( p-1 \right )G_{n},   & \text{ if $n\equiv0\,(mod\,2)$}, b=0,\\
  p^{n-1},& \text{ if $n\equiv1\,(mod\,2)$}, b=0,\\
  p^{n-1}-p^{-1}G_{n},      & \text{ if $n\equiv0\,(mod\,2)$}, b\ne0,  \\
  p^{n-1}+p^{-1}G_{n}G_1  \eta^{\prime } \left (  b \right ),& \text{ if $n\equiv1\,(mod\,2)$}, b\ne 0,
\end{cases}
    \end{split}
\]
If $\frac{n}{v}\equiv0 \,(mod\,2)$ , by Lemma \ref{lem2}, then
\[
   S_{n,r}\left ( -y,0 \right )  =\begin{cases}
  -p^{\frac{n}{2} +  v}, & \text{ if $\frac{n}{v} \equiv0\,(mod\,4)$}, \\
  -p^{\frac{n}{2} }, & \text{ if $\frac{n}{v} \equiv2\,(mod\,4)$}.
\end{cases}
\]
So, we have
\[
    \begin{split}
         N(a,b) &=\begin{cases}
   p^{n-1} - p^{\frac{n}{2}  +  v-1}\sum_{y\in \mathbb{F} _{p}^{*} }\zeta_p^{-by},   & \text{ if $\frac{n}{v} \equiv0\,(mod\,4)$}, \\
  p^{n-1}-p^{\frac{n}{2} -1}   \sum_{y\in \mathbb{F} _{p}^{*} }\zeta_p^{-by},  & \text{ if $\frac{n}{v} \equiv2\,(mod\,4)$},
\end{cases}\\
&=\begin{cases}
   p^{n-1} -(p-1)p^{\frac{n}{2} +  v-1} ,    & \text{ if $\frac{n}{v} \equiv0\,(mod\,4)$},b= 0,\\
   p^{n-1} +p^{\frac{n}{2} +  v-1},     & \text{ if $\frac{n}{v} \equiv0\,(mod\,4)$},b\ne 0,  \\
   p^{n-1}-(p-1) p^{\frac{n}{2} -1},   & \text{ if $\frac{n}{v} \equiv2\,(mod\,4)$},b= 0,\\
   p^{n-1}+ p^{\frac{n}{2} -1},   & \text{ if $\frac{n}{v} \equiv2\,(mod\,4)$},b\ne 0.
\end{cases}
    \end{split}
\]
When $a\ne0$, if $\frac{n}{v}\equiv 0(\mod4)$, the equation
\begin{align}\label{equal1}
X^{p^{2r} }+  X=- a^{p^{r} }
\end{align}
is not always solvable over $\mathbb{F}_{p^{n}}$. When $X^{p^{2r}}+X$ is a permutation polynomial over $\mathbb{F}_{p^{n}}$, it has a unique solution. Thus, let $\gamma_a$ be some solution of (\ref{equal1}) if exists. For $z_{1},z_{2}\in \mathbb{F}_{p}^{*}$, the equation
\begin{align}\label{equal2}
z_{1}^{p^{2r} }  X^{p^{2r} }+  z_{1} X=- (z_{2}a )^{p^{r} }
\end{align}
always has a solution $z_{3}\gamma _{a}$, where $z_{3}=z_{1}^{-1}z_{2}\in \mathbb{F}_{p}^{*}$. By (\ref{equal1})-(\ref{equal2}), for any $y\in \mathbb{F}_{p}^{*}$, the equation
\[(-y)^{p^{2r} }  X^{p^{2r} }+  (-y) X=- (ya )^{p^{r} }\]
has a solution $x_{0}= -\gamma _{a}$.
Thus, if $\frac{n}{v}\equiv1 \,(mod\,2)$ is odd, by lemma \ref{lem1}, $S_{n,r} \left ( -y,ya \right ) =G_{n}  \eta  \left ( -y \right )\zeta_p ^{Tr \left ( y\gamma_a^{p^{r}+  1 }  \right ) }$, denote $\lambda_a=Tr\left ( \gamma_a^{p^{r}+  1 } \right )$, we have
\[
   \begin{split}
        N(a,b) &=p^{n-1} +  p^{-1}\textstyle \sum_{y\in \mathbb{F} _{p}^{*} }\zeta_p^{-yb} G_{n}  \eta  \left ( -y \right )\zeta_p^{y\lambda_a}\\
&=p^{n-1}+ p^{-1} G_{n}   \textstyle \sum_{y\in \mathbb{F} _{p}^{*} }\eta (-y)\zeta_p^{y(\lambda_a-b)}   \\
&=\begin{cases}
  p^{n-1}+ p^{-1}G_{n}   \textstyle \sum_{y\in \mathbb{F} _{p}^{*} }\zeta_p^{y(\lambda_a-b)}, & \text{ if $n\equiv0\,(mod\,2)$}, \\
  p^{n-1}+ p^{-1}G_{n}   \textstyle \sum_{y\in \mathbb{F} _{p}^{*} }\eta^{\prime } (-y) \zeta_p^{y(\lambda_a-b)},& \text{ if $n\equiv1\,(mod\,2)$},\\
\end{cases}\\
&=\begin{cases}
  p^{n-1}+p^{-1}(p-1)G_{n},  & \text{ if $n\equiv0\,(mod\,2)$},\lambda_a = b, \\
  p^{n-1}-p^{-1}G_{n}, & \text{ if $n\equiv0\,(mod\,2)$},\lambda_a \ne  b,  \\
  p^{n-1},& \text{ if $n\equiv1\,(mod\,2)$},\lambda_a = b ,\\
  p^{n-1}+p^{-1}G_{n}  G_1\eta^{\prime }  (b-\lambda_a ),& \text{ if $n\equiv1\,(mod\,2)$},\lambda_a \ne  b.
\end{cases}
    \end{split}
\]
If $\frac{n}{v}\equiv0 \,(mod\,2)$, $f\left ( X \right ) =\left ( -y \right ) ^{p^{r} } X^{p^{2r} } -yX$ is not always a permutation polynomial, by Lemma \ref{lem1} and Lemma \ref{lem5}, we have
\[
    \begin{split}
        N\left (a, b \right )&=p^{n-1}+ p^{-1} \sum_{y\in \mathbb{F} _{p}^{*} }\zeta_p^{ -by }S_{n,r} \left ( -y,ay \right )\\
        &=\begin{cases}
  p^{n-1} -p^{-1}\sum_{y\in \mathbb{F} _{p}^{*} } \left ( -1\right )  ^{\frac{n}{2v} } p^{\frac{n}{2}+  v } \zeta_p^{y(\lambda_a-b) },  & \text{ if $\frac{n}{v} \equiv0\,(mod\,4)$} ,(\ref{equal1})\ is\  solvable , \\
p^{n-1}, & \text{ if $\frac{n}{v} \equiv0\,(mod\,4)$}, (\ref{equal1})\ is\  unsolvable,\\
  p^{n-1} + p^{-1} \sum_{y\in \mathbb{F} _{p}^{*} }  \left ( -1\right )  ^{\frac{n}{2v} } p^{\frac{n}{2} } \zeta_p^{y(\lambda_a-b) }, & \text{ if $\frac{n}{v} \equiv2\,(mod\,4)$},
\end{cases} \\
&= \begin{cases}
p^{n-1}-(p-1)  p^{\frac{n}{2}+  v -1},  & \text{ if $\frac{n}{v} \equiv0\,(mod\,4)$},\lambda_a = b, (\ref{equal1})\ is\  solvable,\\
  p^{n-1}+  p^{\frac{n}{2}+  v -1},  & \text{ if $\frac{n}{v} \equiv0\,(mod\,4)$},\lambda_a \ne b, (\ref{equal1})\ is\  solvable,\\
  p^{n-1}, & \text{ if $\frac{n}{v} \equiv0\,(mod\,4)$}, (\ref{equal1})\ is\  unsolvable,\\
  p^{n-1} -(p-1)  p^{\frac{n}{2}-1 },& \text{ if $\frac{n}{v} \equiv2\,(mod\,4)$}, \lambda_a =b,\\
 p^{n-1} +  p^{\frac{n}{2}-1 },& \text{ if $\frac{n}{v} \equiv2\,(mod\,4)$},\lambda_a \ne b.
\end{cases}
    \end{split}
\]
This completes the proof of Lemma \ref{lem10}.
\end{proof}
Based on Lemma \ref{lem10}, the $P_{I}$ of the authentication code in (\ref{equal5}) can be obtained.
\begin{theorem}\label{thm1}
Let $n$ and $r$ be two positive integers and $v=gcd(n,r)$, for the authentication code in (\ref{equal5}), we have the following,
\[
P_{I} =\begin{cases}
 \frac{1}{p}+  \frac{1}{p^{\frac{n+1}{2}  } }, & \text{ if $n\equiv1\,(mod\,2)$ },\\
\frac{1}{p}+  \frac{p-1}{p^{ \frac{n}{2}+1 } },  & \text{ if $n\equiv0\,(mod\,2)$ and $\frac{n}{v}\equiv1 \,(mod\,2)$ },   \\
 \frac{1}{p}+  \frac{1}{p^{\frac{n}{2}+1 }  }, &  \text{ if $\frac{n}{v} \equiv2\,(mod\,4)$} ,   \\
  \frac{1}{p}+   \frac{1 }{p^{\frac{n}{2}-v+1} }, &  \text{ if $\frac{n}{v} \equiv0\,(mod\,4)$}.
\end{cases}
\]
\end{theorem}

\begin{proof}
By (\ref{euqal6}), we know
\[P_{I} =\max_{a\in \mathbb{F}_{p^{n} },b\in \mathbb{F}_{p}    }\frac{N(a,b)}{p^{n} }.\]
By Lemma \ref{lem10},
if $n\equiv1\,(mod\,2)$ , we have
\[
   \frac{N(a,b)}{p^{n} } =\begin{cases}
   \frac{1}{p}, & \text{ if }  a=0,b=0 \ or \ a\ne0,\lambda_a =b,\\
  \frac{1}{p}\pm \frac{1}{p^{\frac{n+1}{2} } } ,   & \text{ if } a=0,b\ne0 \ or \ a\ne 0,\lambda_a \ne b;
\end{cases}
\]
if $n\equiv0\,(mod\,2)$ and $\frac{n}{v}\equiv1 \,(mod\,2)$ , we have
\[
     \frac{N(a,b)}{p^{n} } =\begin{cases}
  \frac{1}{p}\pm \frac{p-1}{p^{ \frac{n}{2}+1 } },   & \text{ if }a=0,b= 0\ or\ a\ne0,\lambda_a = b,\\
  \frac{1}{p}\mp \frac{1}{p^{\frac{n}{2}+1 } } ,  & \text{ if }a=0,b\ne 0\ or\ a\ne0,\lambda_a \ne b;
\end{cases}
\]
if $\frac{n}{v} \equiv2\,(mod\,4)$, we have
\[
     \frac{N(a,b)}{p^{n} }= \begin{cases}
     \frac{1}{p}  - \frac{p-1}{p^{\frac{n}{2}+1  } },  & \text{ if }a=0,b= 0\ or\ a\ne 0, \lambda_a = 0, \\
   \frac{1}{p}  + \frac{1}{p^{ \frac{n}{2}+1} } , & \text{ if }a=0,b\ne 0\ or\ a\ne 0, \lambda_a \ne 0.
\end{cases}
\]
if $\frac{n}{v} \equiv0\,(mod\,4)$, we have
\[
    \begin{split}
       \frac{N(a,b)}{p^{n} }= \begin{cases}
   \frac{1}{p}  -\frac{p-1 }{p^{\frac{n}{2}-v+  1} },    & \text{ if }a=0,b= 0\ or\ a\ne 0,\lambda_a = b, (\ref{equal1})\ is\  solvable,  \\
   \frac{1}{p}  +  \frac{1 }{p^{\frac{n}{2}-v+  1} } ,   & \text{ if }a=0,b\ne 0\ or\ a\ne 0,\lambda_a \ne b, (\ref{equal1})\ is\  solvable , \\
   \frac{1}{p},& \text{ if }a\ne 0, (\ref{equal1})\ is\  unsolvable.
\end{cases}
    \end{split}
\]
It follows that
\[ P_{I} =\begin{cases}
 \frac{1}{p}+  \frac{1}{p^{\frac{n+1}{2}  } } ,& \text{ if $n\equiv 1\,(mod\,2)$ } ,\\
\frac{1}{p}+  \frac{p-1}{p^{  \frac{n}{2}+1 } } , &  \text{ if $n\equiv0\,(mod\,2)$ and $\frac{n}{v}\equiv1 \,(mod\,2)$ },   \\
 \frac{1}{p}+  \frac{1}{p^{\frac{n}{2}+1 }  }, &  \text{ if $\frac{n}{v} \equiv2\,(mod\,4)$},   \\
 \frac{1}{p}+   \frac{1 }{p^{\frac{n}{2}-v+1} } , & \text{ if $\frac{n}{v} \equiv0\,(mod\,4)$} .
\end{cases} \]
 This completes the proof of Theorem \ref{thm1}.
\end{proof}
\subsection{Substitution attack}

Suppose an opponent has observed one message $m =\left ( m_{1} =s+  k^{p^{r} } ,m_{2}=Tr(sk)  \right ) $. Now, he wants to replace $m$ with another message $m^{\prime} =\left ( m_{1}^{\prime} ,m_{2}^{\prime} \right ) $, where $m_{1} ^{\prime} \ne m_{1} $. Let $\sigma_{1}  =m_{1} ^{\prime}- m_{1} $ and $\sigma_{2}  =m_{2} ^{\prime}- m_{2} $, then substituting $m$ with $m^{\prime}$ is equivalent to adding an element $\sigma_{1}$ to $m_{1}$,  and add an element $\sigma_{2}$ to $m_{2}$. This is successful if and only if $Tr\left ( sk \right )+  \sigma _{2}  =Tr\left ( \left ( s+  \sigma _{1}  \right ) k \right ) $, which is equivalent to $Tr\left ( \sigma _{1}k  \right ) =\sigma _{2} $. Hence, we have
\[
    \begin{split}
       P_{S} &=\max_{\sigma _{1}\ne 0,m_{1}\in\mathbb{F}_{p^n}, m_{2},\sigma _{2}\in\mathbb{F}_p  } \frac{\left | \left \{ k\in \mathbb{F} _{p^{n} } |Tr\left ( \left ( m_{1}-k^{p^{r} }   \right )k  \right )=m_{2},Tr\left ( \sigma _{1}k  \right ) =\sigma _{2}     \right \}  \right | }{\left | \left \{ k\in \mathbb{F} _{p^{n} }|Tr\left ( \left ( m_{1}-k^{p^{r} }   \right )k  \right )=m_{2}  \right \}  \right | }\\
&= \max_{\alpha  \ne0,a\in\mathbb{F}_{p^n}, b,\beta\in\mathbb{F}_p }\frac{\left | \left \{ x\in \mathbb{F} _{p^{n} }|Tr\left (  ax -x^{p^{r}+  1 }  \right )=b ,Tr\left ( \alpha x  \right ) =\beta      \right \}  \right | }{\left | \left \{ x\in \mathbb{F} _{p^{n} }|Tr\left (  ax-x^{p^{r}+  1 }      \right )=b   \right \}  \right | }.
    \end{split}
\]
For $\alpha(\neq 0),a \in \mathbb{F} _{p^{n} }$ and $b , \beta \in \mathbb{F} _{p}$, denote
\[N\left ( a, b, \alpha , \beta  \right ) =\left | \left \{ x\in \mathbb{F} _{p^{n} }|Tr\left (   ax -x^{p^{r}+  1} \right )=b \ \text{and}\ Tr\left ( \alpha x  \right ) =\beta     \right \}  \right |.\]
Next, in order to calculate $N\left ( a, b, \alpha , \beta  \right )$, we need the following lemma.
\begin{lemma}\label{lem13}
Let $n$ and $r$ be positive integers with $\gcd(n,r)=v$, and let $a,\alpha\in \mathbb{F}_{p^n}^\ast$. Assume that $\frac{n}{v} \equiv0\,(mod\,4)$, if neither of the equations
\[X^{p^{2r}}+X=-a^{p^r} \quad \text{nor} \quad X^{p^{2r}}+X=-\alpha^{p^r}\]
is solvable, then for $t=1,2,\ldots,p-1$, at most one of the equations
\[X^{p^{2r}}+X=(a+t\alpha)^{p^r}\]
is solvable.
\end{lemma}
\begin{proof}
We sketch the proof by contradiction. Assume, for contradiction, that there exist two distinct $t_1,t_2\in \{1,2,\ldots,p-1\}$ such that both equations
\[X^{p^{2r}}+X=(a+t_i\alpha)^{p^r}\quad (i=1,2)\]
are solvable in $\mathbb{F}_{p^n}$. Denote $x_i$ be a solution to the $i$-th equation, it is easy to verify that $-(x_1-x_2)$ is a solution to the equation
\[X^{p^{2r}}+X=-(t_1-t_2)\alpha^{p^r},\]
which leads to $-\frac{x_1-x_2}{t_1-t_2}$ is a solution to the equation
\[\quad X^{p^{2r}}+X=-\alpha^{p^r}.\]
This yields a contradiction to the assumption that the equation $X^{p^{2r}}+X=-\alpha^{p^r}$ is unsolvable. Hence, at most one of the equations
\[X^{p^{2r}}+X=(a+t\alpha)^{p^r}\quad (t=1,2\ldots,p-1)\]
is solvable. This competes the proof of Lemma \ref{lem13}.
\end{proof}
\begin{lemma}\label{lem11}
For $\alpha(\neq 0), a \in \mathbb{F} _{p^{n} }$ and $b , \beta \in \mathbb{F} _{p}$, denote
\[N\left ( a, b, \alpha , \beta  \right ) =\left | \left \{ x\in \mathbb{F} _{p^{n} }|Tr\left (   ax -x^{p^{r}+  1} \right )=b \ \text{and}\ Tr\left ( \alpha x  \right ) =\beta     \right \}  \right |,\]
then, we have the following result.

(1) If $n\equiv1\,(mod\,2)$, then
\begin{eqnarray*}
N\left ( a, b, \alpha , \beta  \right )&=&\left\{\begin{array}{ll}
p^{n-2},&  \quad \text{ if $\lambda_{\alpha}=0$ and $\Delta\neq 0$},\\
&\quad \text{or if $\lambda_{\alpha}=0$ and $\Delta=0$ and $\lambda_a-b=0$},\\
p^{n-2}+\frac{1}{p}G_nG_1\eta^\prime(b-\lambda_a),&\quad \text{ if $\lambda_{\alpha}=0$ and $\Delta=0$ and $\lambda_a-b\neq0$},\\
p^{n-2}+\frac{1}{p^2}(p-1)G_nG_1\eta^\prime(-\lambda_{\alpha}), & \quad \text{ if $\lambda_{\alpha}\neq 0$ and $\Delta=0$},\\
p^{n-2}-\frac{1}{p^2}G_nG_1\eta^\prime(-\lambda_{\alpha}), & \quad \text{if $\lambda_{\alpha}\neq 0$ and $\Delta\neq0$};
\end{array}\right.
\end{eqnarray*}

(2) if $n\equiv0\,(mod\,2)$ and $\frac{n}{v}\equiv1 \,(mod\,2)$, then
\begin{eqnarray*}
N\left ( a, b, \alpha , \beta  \right )&=&\left\{\begin{array}{ll}
p^{n-2},&  \quad \text{ if $\lambda_{\alpha}=0$ and $\Delta\neq 0$},\\
 & \quad \text{or if $\lambda_{\alpha}\neq 0$ and $\Delta=0$},\\
p^{n-2}+\frac{1}{p}(p-1)G_n, &\quad \text{if $\lambda_{\alpha}=0$ and $\Delta=0$ and $\lambda_a-b=0$},\\
p^{n-2}-\frac{1}{p}G_n, &\quad \text{ if $\lambda_{\alpha}=0$ and $\Delta=0$ and $\lambda_a-b\neq 0$},\\
p^{n-2}+\frac{1}{p^2}G_nG_1^2\eta^\prime(\Delta), & \quad \text{if $\lambda_{\alpha}\neq0$ and $\Delta\neq0$};
\end{array}\right.
\end{eqnarray*}

(3) if $\frac{n}{v} \equiv2\,(mod\,4)$, then
\begin{eqnarray*}
N\left ( a, b, \alpha , \beta  \right )&=&\left\{\begin{array}{ll}
p^{n-2},&  \quad \text{ if $\lambda_{\alpha}=0$ and $\Delta\neq 0$},\\
 & \quad \text{or if $\lambda_{\alpha}\neq0$ and $\Delta=0$},\\
p^{n-2}-(p-1)p^{\frac{n}{2}-1}, &\quad \text{if $\lambda_{\alpha}=0$ and $\Delta=0$ and $\lambda_a-b=0$},\\
p^{n-2}+p^{\frac{n}{2}-1}, &\quad \text{if $\lambda_{\alpha}=0$ and $\Delta=0$ and $\lambda_a-b\neq 0$},\\
p^{n-2}-p^{\frac{n}{2}-2}G_1^2\eta^\prime(\Delta), & \quad \text{if $\lambda_{\alpha}\neq0$ and $\Delta\neq0$};
\end{array}\right.
\end{eqnarray*}

(4)if $\frac{n}{v} \equiv0\,(mod\,4)$, then
\begin{eqnarray*}
N\left ( a, b, \alpha , \beta  \right )&=&\left\{\begin{array}{ll}
p^{n-2}, & \quad \text{if (\ref{equal1}) and (\ref{equal7}) are solvable and $\lambda_{\alpha}\neq0$ and $\Delta=0$},\\
& \quad \text{or if (\ref{equal1}) and (\ref{equal7}) are unsolvable},\\
&  \quad \text{or if (\ref{equal1}) and (\ref{equal7}) are solvable and $\lambda_{\alpha}=0$ and $\Delta\neq 0$},\\
p^{n-2}-(p-1)p^{\frac{n}{2}+v-1}, &  \quad \text{if (\ref{equal1}) and (\ref{equal7}) are solvable and $\lambda_{\alpha}=0$ and  $\Delta=0$ and $\lambda_a-b=0$},\\
p^{n-2}+p^{\frac{n}{2}+v-1}, &  \quad \text{if (\ref{equal1}) and (\ref{equal7}) are solvable and $\lambda_{\alpha}=0$ and $\Delta=0$ and $\lambda_a-b\neq 0$},\\
p^{n-2}-p^{\frac{n}{2}+v-2}G_1^2\eta^\prime(\Delta), & \quad \text{if (\ref{equal1}) and (\ref{equal7}) are solvable and $\lambda_{\alpha}\neq0$ and $\Delta\neq 0$},\\
p^{n-2}-(p-1)p^{\frac{n}{2}+v-2}, & \quad \text{if (\ref{equal1}) is solvable and (\ref{equal7}) is unsolvable and $\lambda_a-b=0$},\\
& \quad \text{or if and (\ref{equal1}) is unsolvable and (\ref{equal7}) is solvable and $\lambda_{t_0}-b-t_0\beta=0$},\\
p^{n-2}+p^{\frac{n}{2}+v-2}, & \quad \text{if and (\ref{equal1}) is solvable and (\ref{equal7}) is unsolvable and $\lambda_a-b\neq0$},\\
&\quad \text{or if and (\ref{equal1}) is unsolvable and (\ref{equal7}) is solvable and $\lambda_{t_0}-b-t_0\beta\neq0$};\\
\end{array}\right.
\end{eqnarray*}
Here $\lambda_a = \operatorname{Tr}(\gamma_a^{p^r+1})$, $\lambda_\alpha = \operatorname{Tr}(\gamma_\alpha^{p^r+1})$, $\lambda_{\alpha a} = \operatorname{Tr}(\gamma_\alpha^{p^r}\gamma_a+\gamma_a^{p^r}\gamma_\alpha)$, and $\Delta = 4\lambda_\alpha(\lambda_a-b) - (\lambda_{\alpha a} - \beta)^2$, where $\gamma_a$ and $\gamma_\alpha$ are the solutions of (\ref{equal1}) and (\ref{euqal4}), respectively. Moreover, $\lambda_{t_0} = \operatorname{Tr}(\gamma_{t_0}^{p^r+1})$, with $\gamma_{t_0}$ being the solution of (\ref{equal7}) corresponding to the unique value $t_0 = \frac{y_2}{y_1}\in\mathbb{F}_p^\ast$ for which (\ref{equal7}) is solvable.
\end{lemma}

\begin{proof}
By the orthogonal property of additive character, we have
\[\begin{split}
        N\left ( a,b,\alpha ,\beta  \right ) &=\frac{1}{p^2}\sum_{x\in \mathbb{F} _{p^{n} } }\sum_{y_{1}\in \mathbb{F}_{p}}\zeta _{p}^{y_{1}(Tr(ax-x^{p^{r}+ 1})-b) }\sum_{y_{2}\in \mathbb{F}_{p}}\zeta _{p}^{y_{2}(Tr(\alpha x)-\beta ) } \\
&= \frac{1}{p^2} \sum_{x\in \mathbb{F} _{p^{n} }  }(1+\sum_{y_{1}\in \mathbb{F}^{*} _{p}}\zeta _{p}^{y_{1}(Tr(ax-x^{p^{r}+ 1})-b) })(1+\sum_{y_{2}\in \mathbb{F}^{*} _{p}}\zeta _{p}^{y_{2}(Tr(\alpha x)-\beta ) })\\
&=p^{n-2}+\frac{1}{p^2}\sum_{y_{2}\in \mathbb{F}^{*}_{p} }\zeta _{p}^{-y_{2}\beta}\sum_{x\in \mathbb{F}_{p^{n}} }\zeta _{p}^{y_{2}Tr(\alpha x)}+\frac{1}{p^2}\sum_{y_{1}\in \mathbb{F}^{*}_{p} }\zeta _{p}^{-y_{1}b}\sum_{x\in \mathbb{F}_{p^{n}} }\zeta _{p}^{Tr(-y_{1}x^{p^{r}+1}+y_{1}ax)}\\
&+\frac{1}{p^2}\sum_{y_{1},y_{2}\in \mathbb{F}^{*}_{p}  }\zeta _{p}^{-y_{1}b-y_{2}\beta }\sum_{x\in \mathbb{F} _{p^{n} } }\zeta _{p}^{Tr(-y_{1}x^{p^{r}+1}+y_{1}ax+y_{2}\alpha x)}    \\
&=p^{n-2}+\frac{1}{p^2}\sum_{y_{1}\in \mathbb{F}^{*}_{p} }\zeta _{p}^{-y_{1}b}S_{n,r}(-y_{1},y_{1}a)+\frac{1}{p^2}\sum_{y_{1},y_{2}\in \mathbb{F}^{*}_{p}  }\zeta _{p}^{-y_{1}b-y_{2}\beta }S_{n,r}(-y_{1},y_{1}a+y_{2}\alpha )\\
&=p^{n-2}+\frac{1}{p^2}(\Omega_{\alpha,a}),
\end{split}
\]
where $\Omega_{\alpha,a}=\sum_{y_{1}\in \mathbb{F}^{*}_{p} }\zeta _{p}^{-y_{1}b}S_{n,r}(-y_{1},y_{1}a)+\sum_{y_{1},y_{2}\in \mathbb{F}^{*}_{p}  }\zeta _{p}^{-y_{1}b-y_{2}\beta }S_{n,r}(-y_{1},y_{1}a+y_{2}\alpha )$. For convenience, let $\lambda_{\alpha}=Tr(\gamma_{\alpha}^{p^r+1})$, where $\gamma_{\alpha}$ is the solution of
\begin{align}\label{euqal4}
X^{p^{2r}}+X=-\alpha^{p^{r}}.
\end{align}
To avoid verbosity, we provide the computational details only for the case $ \lambda_{\alpha}\neq 0$, omitting the case $\lambda_{\alpha}=0$, whose computation is straightforward.

When $a=0$, $\Omega_{\alpha,0}=\sum_{y_{1}\in \mathbb{F}^{*}_{p} }\zeta _{p}^{-y_{1}b}S_{n,r}(-y_{1},0)+\sum_{y_{1},y_{2}\in \mathbb{F}^{*}_{p}  }\zeta _{p}^{-y_{1}b-y_{2}\beta }S_{n,r}(-y_{1},y_{2}\alpha )$, and denote $\gcd(n,r)=v$, we have following cases.

{\bf Case 1} If  $\frac{n}{v}\equiv1 \,(mod\,2)$, by Lemmas \ref{lem2} and \ref{lem1}, we have
\begin{eqnarray*}
\Omega_{\alpha,0}&=&\sum_{y_{1}\in \mathbb{F}^{*}_{p}}\zeta _{p}^{-y_{1}b}G_n\eta(-y_1)+\sum_{y_{1},y_{2}\in \mathbb{F}^{*}_{p}}\zeta _{p}^{-y_{1}b-y_{2}\beta}G_n\eta(-y_1)\zeta _{p}^{\frac{y_2^2}{y_1}Tr(\gamma_{\alpha}^{p^r+1})}\\
&=&G_n\sum_{y_{1}\in \mathbb{F}^{*}_{p}}\zeta _{p}^{-y_{1}b}\eta(-y_1)+G_n\sum_{y_{1}\in \mathbb{F}^{*}_{p}}\eta(-y_1)\zeta _{p}^{-y_{1}b}\sum_{y_2\in \mathbb{F}^{*}_{p}}\zeta _{p}^{\frac{y_2^2}{y_1}Tr(\gamma_{\alpha}^{p^r+1})-y_2\beta}.
\end{eqnarray*}
By Lemma \ref{lem12}, we obtain that
\begin{eqnarray*}
\Omega_{\alpha,0}&=&G_n\sum_{y_{1}\in \mathbb{F}^{*}_{p}}\zeta _{p}^{-y_{1}b}\eta(-y_1)+G_n\sum_{y_{1}\in \mathbb{F}^{*}_{p}}\eta(-y_1)\zeta _{p}^{-y_{1}b}\sum_{y_2\in \mathbb{F}^{*}_{p}}\zeta _{p}^{\frac{\lambda_{\alpha}}{y_1}y_2^2-\beta y_2}\\
&=&\left\{\begin{array}{ll}
G_nG_1\sum_{y_{1}\in \mathbb{F}^{*}_{p}}\eta^\prime(\frac{\lambda_{\alpha}}{y_1})\zeta _{p}^{\frac{(-4\lambda_{\alpha}b-\beta^2)y_1}{4\lambda_{\alpha}}}, & \quad \text{if $n\equiv0\,(mod\,2)$ },\\
G_nG_1\sum_{y_{1}\in \mathbb{F}^{*}_{p}}\eta^\prime(-\lambda_{\alpha})\zeta _{p}^{\frac{(-4\lambda_{\alpha}b-\beta^2)y_1}{4\lambda_{\alpha}}}, & \quad \text{if $n\equiv1\,(mod\,2)$ },
\end{array}\right.\\
\end{eqnarray*}
\begin{eqnarray*}
&=&\left\{\begin{array}{ll}
0, & \quad \text{if $n\equiv0\,(mod\,2)$ and $-4\lambda_{\alpha}b-\beta^2=0$},\\
G_nG_1^2\eta^\prime(-4\lambda_{\alpha}b-\beta^2), & \quad \text{if $n\equiv0\,(mod\,2)$  and $-4\lambda_{\alpha}b-\beta^2\neq0$},\\
(p-1)G_nG_1\eta^\prime(-\lambda_{\alpha}), & \quad \text{if $n\equiv1\,(mod\,2)$ and $-4\lambda_{\alpha}b-\beta^2=0$},\\
-G_nG_1\eta^\prime(-\lambda_{\alpha}), & \quad \text{if $n\equiv1\,(mod\,2)$ and $-4\lambda_{\alpha}b-\beta^2\neq0$}.
\end{array}\right.
\end{eqnarray*}

{\bf Case 2} If  $\frac{n}{v}\equiv0 \,(mod\,2)$ , then by Lemmas \ref{lem1} and \ref{lem5}, we have
\begin{eqnarray*}
\Omega_{\alpha,0}&=&\left\{\begin{array}{ll}
-p^{\frac{n}{2}}\sum_{y_{1}\in \mathbb{F}^{*}_{p}}\zeta _{p}^{-y_{1}b}-p^{\frac{n}{2}}\sum_{y_{1}\in \mathbb{F}^{*}_{p}}\zeta _{p}^{-y_{1}b}\sum_{y_2\in \mathbb{F}^{*}_{p}}\zeta _{p}^{\frac{\lambda_{\alpha}}{y_1}y_2^2-\beta y_2}, & \quad \text{$\frac{n}{v} \equiv2\,(mod\,4)$},\\
-p^{\frac{n}{2}+v}\sum_{y_{1}\in \mathbb{F}^{*}_{p}}\zeta _{p}^{-y_{1}b}-p^{\frac{n}{2}+v}\sum_{y_{1}\in \mathbb{F}^{*}_{p}}\zeta _{p}^{-y_{1}b}\sum_{y_2\in \mathbb{F}^{*}_{p}}\zeta _{p}^{\frac{\lambda_{\alpha}}{y_1}y_2^2-\beta y_2}, & \quad \text{$\frac{n}{v} \equiv0\,(mod\,4)$ and (\ref{euqal4}) is solvable},\\
-p^{\frac{n}{2}+v}\sum_{y_{1}\in \mathbb{F}^{*}_{p}}\zeta _{p}^{-y_{1}b}, & \quad \text{$\frac{n}{v} \equiv0\,(mod\,4)$ and (\ref{euqal4}) is unsolvable}.
\end{array}\right.
\end{eqnarray*}
By Lemma \ref{lem12}, we can obtain
\begin{eqnarray*}
\Omega_{\alpha,0}&=&\left\{\begin{array}{ll}
-p^{\frac{n}{2}}G_1\sum_{y_{1}\in \mathbb{F}^{*}_{p}}\eta^\prime(\frac{\lambda_{\alpha}}{y_1})\zeta _{p}^{\frac{(-4\lambda_{\alpha}b-\beta^2)y_1}{4\lambda_{\alpha}}}, & \quad \text{if $\frac{n}{v} \equiv2\,(mod\,4)$},\\
-p^{\frac{n}{2}+v}G_1\sum_{y_{1}\in \mathbb{F}^{*}_{p}}\eta^\prime(\frac{\lambda_{\alpha}}{y_1})\zeta _{p}^{\frac{(-4\lambda_{\alpha}b-\beta^2)y_1}{4\lambda_{\alpha}}}, & \quad \text{if $\frac{n}{v} \equiv0\,(mod\,4)$ and (\ref{euqal4}) is solvable},\\
-p^{\frac{n}{2}+v}\sum_{y_{1}\in \mathbb{F}^{*}_{p}}\zeta _{p}^{-y_{1}b}, & \quad \text{if $\frac{n}{v} \equiv0\,(mod\,4)$ and (\ref{euqal4}) is unsolvable}.
\end{array}\right.
\end{eqnarray*}
\begin{eqnarray*}
\Omega_{\alpha,0}&=&\left\{\begin{array}{ll}
0, & \quad \text{if $\frac{n}{v} \equiv2\,(mod\,4)$ and $-4\lambda_{\alpha}b-\beta^2=0$},\\
-p^{\frac{n}{2}}G_1^2\eta^\prime(-4\lambda_{\alpha}b-\beta^2), & \quad \text{if $\frac{n}{v} \equiv2\,(mod\,4)$ and $-4\lambda_{\alpha}b-\beta^2\neq0$},\\
0, & \quad \text{if $\frac{n}{v} \equiv0\,(mod\,4)$ and (\ref{euqal4}) is solvable and $-4\lambda_{\alpha}b-\beta^2=0$},\\
-p^{\frac{n}{2}+v}G_1^2\eta^\prime(-4\lambda_{\alpha}b-\beta^2), & \quad \text{if $\frac{n}{v} \equiv0\,(mod\,4)$ and (\ref{euqal4}) is solvable and $-4\lambda_{\alpha}b-\beta^2\neq0$},\\
-(p-1)p^{\frac{n}{2}+v}, & \quad \text{if $\frac{n}{v} \equiv0\,(mod\,4)$ and (\ref{euqal4}) is unsolvable and $b=0$},\\
p^{\frac{n}{2}+v}, & \quad \text{if $\frac{n}{v} \equiv0\,(mod\,4)$ and (\ref{euqal4}) is unsolvable and $b\neq 0$}.\\
\end{array}\right.
\end{eqnarray*}

When $a\neq0$, suppose the equations (\ref{equal1}) and (\ref{euqal4}) have the solutions $\gamma_a$ and $\gamma_\alpha$, respectively. Then the equation
\begin{eqnarray}\label{equal7}
-y_1X^{p^{2r}}-y_1X=-(y_1a+y_2\alpha)^{p^r}\Longleftrightarrow X^{p^{2r}}+X=(a+\frac{y_2}{y_1}\alpha)^{p^r}
\end{eqnarray}
has a solution $x_0=-\gamma_a-\frac{y_2}{y_1}\gamma_\alpha$, where $y_1,y_2\in\mathbb{F}_p^\ast$. Since $a\neq 0$, $a$ and $\alpha$ may be linearly dependent over $\mathbb{F}_p$, i.e., there exist  $y_1,y_2\in\mathbb{F}_p^\ast$ such that $y_1a+y_2\alpha=0$, equivalently $a=-\frac{y_2}{y_1}\alpha$. In this case, one readily verifies that $\gamma_a=-\frac{y_2}{y_1}\gamma_\alpha$. By Lemmas \ref{lem2}, \ref{lem1} and \ref{lem5}, this linear dependence over $\mathbb{F}_p$ does not effect the computation of $\Omega_{\alpha,a}$, hence this case need not be treated separately. Recall that $\lambda_a=Tr(\gamma_a^{p^r+1})$ and $\lambda_\alpha=Tr(\gamma_\alpha^{p^r+1})$. For notational convenience, set $\lambda_{\alpha a}=Tr(\gamma_\alpha^{p^r}\gamma_a+\gamma_a^{p^r}\gamma_\alpha)$ and $\Delta=4\lambda_\alpha(\lambda_a-b)-(\lambda_{\alpha a}-\beta)^2$, we then obtain the following cases.

{\bf Case 1} If $\frac{n}{v}\equiv1 \,(mod\,2)$, by Lemmas \ref{lem1} and \ref{lem5}, we have
\begin{eqnarray*}
\Omega_{\alpha,\alpha}&=&\sum_{y_{1}\in \mathbb{F}^{*}_{p}}\zeta _{p}^{-y_{1}b}G_n\eta(-y_1)\zeta _{p}^{\lambda_ay_1}+\sum_{y_{1},y_{2}\in \mathbb{F}^{*}_{p}}\zeta _{p}^{-y_{1}b-y_{2}\beta}G_n\eta(-y_1)\zeta _{p}^{\frac{\lambda_\alpha}{y_1}y_2^2+\lambda_{\alpha a}y_2+\lambda_ay_1}\\
&=&G_n\sum_{y_{1}\in \mathbb{F}^{*}_{p}}\eta(-y_1)\zeta _{p}^{(\lambda_a-b)y_1}\sum_{y_2\in \mathbb{F}_{p}}\zeta _{p}^{\frac{\lambda_\alpha}{y_1}y_2^2+(\lambda_{\alpha a}-\beta)y_2}\\
&=&\left\{\begin{array}{ll}
G_nG_1\sum_{y_{1}\in \mathbb{F}^{*}_{p}}\eta^\prime(\frac{\lambda_{\alpha}}{y_1})\zeta _{p}^{\frac{\Delta y_1}{4\lambda_{\alpha}}}, & \quad \text{if $n\equiv0\,(mod\,2)$},\\
G_nG_1\sum_{y_{1}\in \mathbb{F}^{*}_{p}}\eta^\prime(-\lambda_{\alpha})\zeta _{p}^{\frac{\Delta y_1}{4\lambda_{\alpha}}}, & \quad \text{if $n\equiv1\,(mod\,2)$},
\end{array}\right.\\
\end{eqnarray*}
\begin{eqnarray*}
&=&\left\{\begin{array}{ll}
0, & \quad \text{if $n\equiv0\,(mod\,2)$ and $\Delta=0$},\\
G_nG_1^2\eta^\prime(\Delta), & \quad \text{if $n\equiv0\,(mod\,2)$ and $\Delta\neq0$},\\
(p-1)G_nG_1\eta^\prime(-\lambda_{\alpha}), & \quad \text{if $n\equiv1\,(mod\,2)$ and $\Delta=0$},\\
-G_nG_1\eta^\prime(-\lambda_{\alpha}), & \quad \text{if $n\equiv1\,(mod\,2)$ and $\Delta\neq0$}.
\end{array}\right.
\end{eqnarray*}

{\bf Case 2} If  $\frac{n}{v}\equiv0 \,(mod\,2)$ . By Lemmas \ref{lem1} and \ref{lem12}, we have

(1) when $\frac{n}{v} \equiv2\,(mod\,4)$,
\begin{eqnarray*}
\Omega_{\alpha,\alpha}&=&-p^{\frac{n}{2}}\sum_{y_{1}\in \mathbb{F}^{*}_{p}}\zeta _{p}^{(\lambda_a-b)y_1}-p^{\frac{n}{2}}\sum_{y_{1}\in \mathbb{F}^{*}_{p}}\zeta _{p}^{(\lambda_a-b)y_{1}}\sum_{y_{2}\in \mathbb{F}^{*}_{p}}\zeta _{p}^{\frac{\lambda_\alpha}{y_1}y_2^2+(\lambda_{\alpha a}-\beta)y_2}\\
&=&-p^{\frac{n}{2}}\sum_{y_{1}\in \mathbb{F}^{*}_{p}}\zeta _{p}^{(\lambda_a-b)y_{1}}\sum_{y_{2}\in \mathbb{F}_{p}}\zeta _{p}^{\frac{\lambda_\alpha}{y_1}y_2^2+(\lambda_{\alpha a}-\beta)y_2}\\
&=&-p^{\frac{n}{2}}G_1\sum_{y_{1}\in \mathbb{F}^{*}_{p}}\eta^\prime(\frac{\lambda_{\alpha}}{y_1})\zeta _{p}^{\frac{\Delta y_1}{4\lambda_{\alpha}}}\\
&=&\left\{\begin{array}{ll}
0, & \quad \text{if $\Delta=0$},\\
-p^{\frac{n}{2}}G_1^2\eta^\prime(\Delta), & \quad \text{if $\Delta\neq0$}.
\end{array}\right.
\end{eqnarray*}

(2) When $\frac{n}{v} \equiv0\,(mod\,4)$, the equations (\ref{equal1}) and (\ref{euqal4}) may or may not admit solutions, and the solvability of equation (\ref{equal7}) then falls into the following cases. If both (\ref{equal1}) and (\ref{euqal4}) admit solutions, then the equation (\ref{equal7}) is solvable. If exactly one of them admits solutions, then the equation (\ref{equal7}) is unsolvable. If neither admits solutions, then the equation (\ref{equal7}) may or may not be solvable. Moreover, when the equation (\ref{equal7}) is solvable, by Lemma \ref{lem13}, there exists the unique $t_0=\frac{y_2}{y_1}$ for which the equation (\ref{equal7}) is solvable. Denote the corresponding solution by $\gamma_{t_0}$ and set $\lambda_{t_0}=Tr(\gamma_{t_0}^{p^r+1})$, we then have the following.
\begin{eqnarray*}
\Omega_{\alpha,\alpha}&=&\left\{\begin{array}{ll}
-p^{\frac{n}{2}+v}\sum_{y_{1}\in \mathbb{F}^{*}_{p}}\zeta _{p}^{(\lambda_a-b)y_{1}}\sum_{y_{2}\in \mathbb{F}_{p}}\zeta _{p}^{\frac{\lambda_\alpha}{y_1}y_2^2+(\lambda_{\alpha a}-\beta)y_2}, & \quad \text{if (\ref{equal1}) is solvable and (\ref{equal7}) is solvable},\\
-p^{\frac{n}{2}+v}\sum_{y_{1}\in \mathbb{F}^{*}_{p}}\zeta _{p}^{(\lambda_a-b)y_1}, & \quad \text{if (\ref{equal1}) is solvable and (\ref{equal7}) is unsolvable},\\
\sum_{y_{1}\in \mathbb{F}^{*}_{p}}\zeta _{p}^{-y_{1}b}\sum_{t_1\in \mathbb{F}^{*}_{p}}\zeta _{p}^{-t_1y_{1}\beta}S_{n,r}(-y_{1},y_{1}a+t_1y_{1}\alpha), & \quad \text{if (\ref{equal1}) is unsolvable and (\ref{equal7}) is solvable},\\
0, & \quad \text{if (\ref{equal1}) is unsolvable and (\ref{equal7}) is unsolvable}.
\end{array}\right.\\
\end{eqnarray*}
\begin{eqnarray*}
&=&\left\{\begin{array}{ll}
-p^{\frac{n}{2}+v}G_1\sum_{y_{1}\in \mathbb{F}^{*}_{p}}\eta^\prime(\frac{\lambda_{\alpha}}{y_1})\zeta _{p}^{\frac{\Delta y_1}{4\lambda_{\alpha}}}, & \quad \text{if (\ref{equal1}) is solvable and (\ref{equal7}) is solvable},\\
-p^{\frac{n}{2}+v}\sum_{y_{1}\in \mathbb{F}^{*}_{p}}\zeta _{p}^{(\lambda_a-b)y_1}, & \quad \text{if (\ref{equal1}) is solvable and (\ref{equal7}) is unsolvable},\\
-p^{\frac{n}{2}+v}\sum_{y_{1}\in \mathbb{F}^{*}_{p}}\zeta _{p}^{(\lambda_{t_0}-b-t_0\beta)y_1}, & \quad \text{if (\ref{equal1}) is unsolvable and (\ref{equal7}) is solvable},\\
0, & \quad \text{if (\ref{equal1}) is unsolvable and (\ref{equal7}) is unsolvable}.
\end{array}\right.\\
&=&\left\{\begin{array}{ll}
0, & \quad \text{if (\ref{equal1}) is solvable and (\ref{equal7}) is solvable and $\Delta=0$},\\
-p^{\frac{n}{2}+v}G_1^2\eta^\prime(\Delta), & \quad \text{if (\ref{equal1}) is solvable and (\ref{equal7}) is solvable and $\Delta\neq0$},\\
-(p-1)p^{\frac{n}{2}+v}, & \quad \text{if (\ref{equal1}) is solvable and (\ref{equal7}) is unsolvable and $\lambda_a=b$},\\
p^{\frac{n}{2}+v}, & \quad \text{if (\ref{equal1}) is solvable and (\ref{equal7}) is unsolvable and $\lambda_a\neq b$},\\
-(p-1)p^{\frac{n}{2}+v}, & \quad \text{if (\ref{equal1}) is unsolvable and (\ref{equal7}) is solvable and $\lambda_{t_0}=b+t_0\beta$},\\
p^{\frac{n}{2}+v}, & \quad \text{if (\ref{equal1}) is unsolvable and (\ref{equal7}) is solvable and $\lambda_{t_0}\neq b+t_0\beta$},\\
0, & \quad \text{if (\ref{equal1}) is unsolvable and (\ref{equal7}) is unsolvable}.
\end{array}\right.\\
\end{eqnarray*}
Note that $\Delta=-4\lambda_{\alpha}b-\beta^2$ if $a=0$, then this completes the proof of Lemma \ref{lem11}.
\end{proof}
Recall that
\[N\left ( a, b  \right )=\left | \left \{ x\in \mathbb{F} _{p^{n} }|Tr\left (  ax-x^{p^{r}+  1 }\right )=b   \right \}  \right |,\]
and then
\begin{align}
    P_{S} =\max_{\alpha  \ne0,a,b,\beta  }\frac{N\left ( a,b,\alpha ,\beta  \right ) }{N\left ( a,b  \right ) }.
\end{align}
Base on Lemmas \ref{lem10} and \ref{lem11}, the $P_{S}$ of the authentication code in (\ref{equal5}) will be given by the following theorem.
\begin{theorem}\label{thm2}
Let $n$ and $r$ be two positive integers with $v=gcd(n,r)$, for the authentication code in (\ref{equal5}), we have the following
\begin{eqnarray*}
 P_{S} &=&\left\{\begin{array}{ll}
  \frac{p^{n-2}+(p-1)p^{\frac{n-3}{2}}}{p^{n-1}}, & \text{ if $n\equiv1\,(mod\,2)$},\\
  \frac{p^{n-2}+p^{\frac{n}{2}-1}}{p^{n-1}-p^{\frac{n}{2}-1}} , & \text{ if $n\equiv0\,(mod\,2)$ and $\frac{n}{v}\equiv1 \,(mod\,2)$ }, \\
 \frac{p^{n-2}}{p^{n-1}-(p-1)p^{\frac{n}{2}-1}} ,& \text{ if $\frac{n}{v} \equiv2\,(mod\,4)$}, \\
 \frac{p^{n-2}}{p^{n-1}-(p-1)p^{\frac{n}{2}+v-1}}, & \text{ if $\frac{n}{v} \equiv0\,(mod\,4)$}.
\end{array}\right.\\
\end{eqnarray*}
\end{theorem}
\begin{proof}
Let the notations be defined as before. We check the conditions of Lemmas \ref{lem10} and \ref{lem11}, then obtain the following
\begin{eqnarray*}
 P_{S} &=&\left\{\begin{array}{ll}
  \frac{p^{n-2}+p^{-2}(p-1)G_nG_1\eta^\prime(-\lambda_\alpha)}{p^{n-1}}, & \text{ if $n\equiv1\,(mod\,2)$ and $\lambda_\alpha\neq 0$ $a=0$, $b=0$, $\Delta=0$},\\
  \frac{p^{n-2}+p^{-2}G_nG_1^2\eta^\prime(\Delta)}{p^{n-1}-p^{-1}G_n} , & \text{  if $n\equiv0\,(mod\,2)$ and $\frac{n}{v}\equiv1 \,(mod\,2)$  and $\lambda_\alpha\neq 0$, $\Delta\neq 0$, $a=0$, $b\neq 0$}, \\
 \frac{p^{n-2}}{p^{n-1}-(p-1)p^{\frac{n}{2}-1}} ,& \text{ if $\frac{n}{v} \equiv2\,(mod\,4)$ and $\lambda_\alpha\neq 0$, $\Delta=0$, $a=0$, $b=0$}, \\
 \frac{p^{n-2}}{p^{n-1}-(p-1)p^{\frac{n}{2}+v-1}}, & \text{ if $\frac{n}{v} \equiv0\,(mod\,4)$ and (\ref{equal1}) and (\ref{equal7}) are solvable, $\lambda_\alpha\neq 0$, $a=0$, $b=0$, $\Delta=0$}.
\end{array}\right.\\
\end{eqnarray*}
The desired results follows from Lemma \ref{lem14}. This completes the proof of Theorem \ref{thm2}.
\end{proof}

\subsection{Optimality of the codes}
In this subsection, we will show $P_I$ and $P_S$ of the authentication code in (\ref{equal5}) asymptotically meets the bounds of Lemma \ref{lem8} and Lemma \ref{lem9}.
\begin{theorem}\label{thm3}
Let the notations be defined as in Theorem \ref{thm1}, then the authentication code in (\ref{equal5}) meets the lower bound of Lemma \ref{lem9} asymptotically.
\end{theorem}
\begin{proof} We only give the proof for $n\equiv1\,(mod\,2)$, omit the proof of other cases, whose proof is similar to this case. For convenience, let
 \[
 R=\frac{|\mathscr{S}|}{|\mathscr{M}|}=\frac{p^{n}  }{p^{n}\times p } =\frac{1}{p},\  \text{and}\ P=\frac{|\mathscr{S}|-1}{|\mathscr{M}|-1}=\frac{p^{n}-1  }{p^{n}\times p-1 } =\frac{p^{n}- 1}{p^{n+  1}-1 },
 \]
 by Theorem \ref{thm1} and \ref{thm2}, we have
 \[P_{I}=\frac{1}{p}+\frac{1}{p^{\frac{n+1}{2} } },\  \text{and} \  P_{S}=\frac{p^{n-2}+(p-1)p^{\frac{n-3}{2}}}{p^{n-1}}.\]
 Thus,
 \[\lim_{n \to \infty} \frac{R}{P_{I} } =1\ \text{and}\ \lim_{n \to \infty} \frac{P}{P_{S} } =1.\]
 This completes the proof of Theorem \ref{thm3}.
 \end{proof}

\begin{theorem}\label{thm4}
 Let the notations be defined as in Theorem \ref{thm1}, then the $P_{I}$ of the authentication code in (\ref{equal5}) meets the lower bound of Lemma \ref{lem8} asymptotically.
 \end{theorem}
 \begin{proof}
 We clearly know $H\left ( \mathscr{E}  \right ) =\log_{2}{p^{n} } $. Next, we compute $H(\mathscr{E}|\mathscr{M})$. Since all source states and encoding rules are used with equal probability, suppose that a message $m=(a,b  )$ has been observed, we will compute the number of $(s,k)$ such that $(a,b)=\left ( s+  k^{p^{r} },Tr\left ( sk \right )   \right )$, namely the number of $k$ such that
 $b=Tr(ak-k^{p^{r}+  1})$, then we can obtain the probability distribution of the messages. For convenience, denote
 \[N(a,b)=|\{k\in\mathbb{F}_{p^n}|Tr(ak-k^{p^{r}+  1})=b\}|. \]
Based on Lemma \ref{lem10}, we prove Theorem \ref{thm4} by considering the following cases.

If $n\equiv1\,(mod\,2)$, by Lemma \ref{lem10}, we can get
\[\begin{split}
&N(0,0)=p^{n-1},\\
&N(0,b)=p^{n-1}+p^{\frac{n-1}{2} }\eta ^{\prime }(b),\\
&N(a,0)
=\begin{cases}
p^{n-1} ,  & \text{ if } \lambda_a =0, \\
p^{n-1}+p^{\frac{n-1}{2} }\eta ^{\prime }(-\lambda_a ),    & \text{ if } \lambda_a \ne 0,
\end{cases}\\
&N(a,b)
=\begin{cases}
p^{n-1},   & \text{ if } \lambda_a =b, \\
p^{n-1}+p^{\frac{n-1}{2} }\eta ^{\prime }(b-\lambda_a ),    & \text{ if } \lambda_a \ne b.
\end{cases}
\end{split}
\]
Based on the discussions above and Lemma \ref{lem14}, the following table is obtained.

\begin{table}[H]
  \centering
  \caption{Message distribution for $n\equiv1\,(mod\,2)$}
  \scriptsize
  \setlength{\tabcolsep}{4pt}
  \renewcommand{\arraystretch}{1.35}
  \resizebox{0.853\textwidth}{!}{%
    \begin{tabular}{l|c|c|c}
      Condition on $m=(a,b)$ & Conditional uncertainty & Probability & Number of messages \\
      \hline
      $a=0,\ b=0$
        & $\log_2 p^{n-1}$ & $\dfrac{p^{n-1}}{p^{2n}}$ & $1$\\
      $a=0,\ b\ne0,\ \varepsilon_{p,n}\eta'(b)=1$
        & $\log_2\left(p^{n-1}+p^{\frac{n-1}{2}}\right)$
        & $\dfrac{p^{n-1}+p^{\frac{n-1}{2}}}{p^{2n}}$ & $\dfrac{p-1}{2}$\\
      $a=0,\ b\ne0,\ \varepsilon_{p,n}\eta'(b)=-1$
        & $\log_2\left(p^{n-1}-p^{\frac{n-1}{2}}\right)$
        & $\dfrac{p^{n-1}-p^{\frac{n-1}{2}}}{p^{2n}}$ & $\dfrac{p-1}{2}$\\
      $a\ne0,\ b=\lambda_a$
        & $\log_2 p^{n-1}$ & $\dfrac{p^{n-1}}{p^{2n}}$ & $p^n-1$\\
      $a\ne0,\ b\ne\lambda_a,\ \varepsilon_{p,n}\eta'(b-\lambda_a)=1$
        & $\log_2\left(p^{n-1}+p^{\frac{n-1}{2}}\right)$
        & $\dfrac{p^{n-1}+p^{\frac{n-1}{2}}}{p^{2n}}$ & $\dfrac{(p^n-1)(p-1)}{2}$\\
      $a\ne0,\ b\ne\lambda_a,\ \varepsilon_{p,n}\eta'(b-\lambda_a)=-1$
        & $\log_2\left(p^{n-1}-p^{\frac{n-1}{2}}\right)$
        & $\dfrac{p^{n-1}-p^{\frac{n-1}{2}}}{p^{2n}}$ & $\dfrac{(p^n-1)(p-1)}{2}$\\
    \end{tabular}%
  }
\end{table}
where $ \varepsilon_{p,n}=\frac{G_nG_1}{p^{\frac{n+1}{2}}}\in\{1,-1\}$.
Then, we have
\[
\begin{split}
  H(\mathscr{E}  |\mathscr{M} )&=\frac{1}{p}\log_2 p^{n-1}
+\frac{p-1}{2p^n}
  \left(p^{n-1}+p^{\frac{n-1}{2}}\right)
  \log_2\left(p^{n-1}+p^{\frac{n-1}{2}}\right)
+\frac{p-1}{2p^n}
  \left(p^{n-1}-p^{\frac{n-1}{2}}\right)
  \log_2\left(p^{n-1}-p^{\frac{n-1}{2}}\right).\\
  &=(n-1)\log_2p+\frac{p-1}{2p}\Biggl[\left(1+p^{-\frac{n-1}{2}}\right)\log_2\left(1+p^{-\frac{n-1}{2}}\right)
+\left(1-p^{-\frac{n-1}{2}}\right)
\log_2\left(1-p^{-\frac{n-1}{2}}\right)\Biggr]\\
&=(n-1)\log_2p+O\left(p^{-\frac{n-1}{2}}\right)
\end{split}
\]
and
\[Q=2^{H(\mathscr{E} |\mathscr{M} )-H(\mathscr{E} )}=\frac{p^{n-1}\cdot 2^{O\left(p^{-\frac{n-1}{2}}\right)}}{p^n}\]
By Theorem \ref{thm1}, we have $P_I=\frac{1}{p}+\frac{1}{p^{\frac{n+1}{2}}}$, it follows that
\[\lim_{n \to \infty} \frac{Q}{P_{I}} =1.\]

If $n\equiv0\,(mod\,2)$ and $\frac{n}{v}\equiv1 \,(mod\,2)$
 or $\frac{n}{v} \equiv2\,(mod\,4)$, a similar calculation to that in the first case yields the result
\[\lim_{n \to \infty} \frac{Q}{P_{I}} =1,\]
the computational details are omitted here to save space.

If $\frac{n}{v} \equiv0\,(mod\,4)$, by Lemma \ref{lem10}, we can get
\[\begin{split}
&N(0,0)= p^{n-1}-(p-1)p^{\frac{n}{2}+v-1},\\
&N(0,b)= p^{n-1}+p^{\frac{n}{2}+v-1},\\
&N(a,0)=\begin{cases}
  p^{n-1}-(p-1)p^{\frac{n}{2}+v-1},  & \text{ if } \lambda_a =0,(\ref{equal1})\ is\  solvable ,\\
  p^{n-1}+p^{\frac{n}{2}+v-1},  & \text{ if } \lambda_a \ne 0,(\ref{equal1})\ is\  solvable ,\\
  p^{n-1}, & \text{ if } (\ref{equal1})\ is\  unsolvable,
\end{cases}\\
&N(a,b)=\begin{cases}
  p^{n-1}-(p-1)p^{\frac{n}{2}+v-1},  & \text{ if } \lambda_a =b,(\ref{equal1})\ is\  solvable, \\
  p^{n-1}+p^{\frac{n}{2}+v-1},  & \text{ if } \lambda_a \ne b,(\ref{equal1})\ is\  solvable ,\\
  p^{n-1} ,& \text{ if } (\ref{equal1})\ is\  unsolvable.
\end{cases}
\end{split}
\]
Based on the discussions above, the following table is obtained.
\begin{table}[H]
  \centering
  \caption{Message distribution for $\frac{n}{v} \equiv0\,(mod\,4)$}
  \label{tab:even-message-distribution}
  \scriptsize
  \setlength{\tabcolsep}{4pt}
  \renewcommand{\arraystretch}{1.35}
  \resizebox{\textwidth}{!}{%
    \begin{tabular}{l|c|c|c}
      Condition on $m=(a,b)$ & Conditional uncertainty & Probability & Number of messages \\
      \hline
      $a = 0,b =0$ & $\log_{2}({p^{n-1}-  (p-1)(p^{\frac{n}{2}+v-1} )})  $ & $\dfrac{p^{n-1}-  (p-1)(p^{\frac{n}{2}+v-1} )}{p^{2n}}   $&$1$\\
     $a= 0,b \ne 0$ & $\log_{2}({p^{n-1}+  p^{\frac{n}{2}+v -1}})$ &$\dfrac{p^{n-1}+  p^{\frac{n}{2}+v -1}}{p^{2n}}  $&$p-1$\\
     $a \ne 0,b = \lambda_a,(\ref{equal1})\ is\  solvable$ & $\log_{2}({p^{n-1}-(p-1)p^{\frac{n}{2}+v-1}} )$ &$\dfrac{p^{n-1}-(p-1)p^{\frac{n}{2}+v-1}}{p^{2n} } $&$p^{n-2v}$\\
     $a \ne 0,b \ne \lambda_a,(\ref{equal1})\ is\  solvable$ & $\log_{2}({p^{n-1}+p^{\frac{n}{2}+v-1}} )$ &$\dfrac{p^{n-1}+p^{\frac{n}{2}+v-1} }{p^{2n} }  $&$p^{n-2v} (p-1)$\\
     $a \ne 0,(\ref{equal1})\ is\  unsolvable$ & $\log_{2}{p^{n-1}} $ &$\dfrac{p^{n-1}}{p^{2n} } $&$p(p^{n}-1)-p^{n-2v+1}$\\
    \end{tabular}%
  }
\end{table}

Then, we have
\[
\begin{split}
    H(\mathscr{E}  |\mathscr{M} )&=\frac{p^{n-1}-  (p-1)(p^{\frac{n}{2}+v-1} )}{p^{2n}} \log_{2}({p^{n-1}-  (p-1)(p^{\frac{n}{2}+v-1} )}) +(p-1)\frac{p^{n-1}+  p^{\frac{n}{2} +v-1}}{p^{2n}}\log_{2}({p^{n-1}+  p^{\frac{n}{2}+v -1}}) \\
    &+(p^{n-2v})\frac{p^{n-1}-(p-1)p^{\frac{n}{2}+v-1}}{p^{2n} }\log_{2}({p^{n-1}-(p-1)p^{\frac{n}{2}+v-1}} )+
[p(p^{n}-1)-p^{n-2v+1}]\frac{p^{n-1} }{p^{2n} }\log_{2}({p^{n-1}} )\\
    &+p^{n-2v}(p-1)\frac{p^{n-1}+p^{\frac{n}{2}+v-1}}{p^{2n} }\log_{2}({p^{n-1}+p^{\frac{n}{2}+v-1}} )\\
    &=(n-1)\log_{2}p+\frac{(1+p^{n-2v})(p^{n-1}-(p-1)p^{\frac{n}{2}+v-1})}{p^{2n}}\log_{2}\Bigl( 1 - (p-1)p^{-\frac{n}{2}+v} \Bigr)\\
&+\frac{(p-1)(1+p^{n-2v})(p^{n-1}+p^{\frac{n}{2}+v-1})}{p^{2n}}\log_{2}\Bigl( 1 + p^{-\frac{n}{2}+v} \Bigr)=(n-1)\log_{2}p+O\left(p^{-\frac{n}{2}+v}\right)
\end{split}
\]
and
\[Q=2^{H(\mathscr{E} |\mathscr{M} )-H(\mathscr{E} )}=\frac{p^{n-1}\cdot 2^{O\left(p^{-\frac{n}{2}+v}\right)}}{p^n}\]
By Theorem \ref{thm1}, we have $P_I=\frac{1}{p}+\frac{1}{p^{\frac{n}{2}-v+1}}$, it follows that
\[\lim_{n \to \infty} \frac{Q}{P_{I}} =1.\]
This completes the proof of Theorem \ref{thm4}.
\end{proof}
\begin{theorem}\label{thm5}
Let the notations be defined as in Theorem \ref{thm1}, then the $P_{S}$ of the authentication code in (\ref{equal5}) meets the lower bound of Lemma \ref{lem8} asymptotically.
\end{theorem}
\begin{proof}
From the proof of Theorem \ref{thm4}, the quantity $H(\mathscr{E}|\mathscr{M})$ is known. We next compute $H(\mathscr{E}|\mathscr{M}^2)$. Since all source states and encoding rules are used with equal probability, suppose that two distinct messages $m_1=(a_1,b_1)$ and $m_2=(a_2,b_2 )$ are observed, both generated by the same key $k$, we shall count the number of pairs $(s_1,k)$ and $(s_2,k)$ satisfying $(a_1,b_1)=\left ( s_1+  k^{p^{r} },Tr\left ( s_1k \right )\right )$ and $(a_2,b_2)=\left ( s_2+  k^{p^{r} },Tr\left ( s_2k \right )\right )$ with $s_1\neq s_2$. Namely the number of $k$ such that
\begin{eqnarray*}
\left\{\begin{array}{ll}
  b_1=Tr(a_1k-k^{p^{r}+  1})\\
  b_2=Tr(a_2k-k^{p^{r}+  1})
\end{array}\right. \Longleftrightarrow
\left\{\begin{array}{ll}
  b_1=Tr(a_1k-k^{p^{r}+  1})\\
  b_1-b_2=Tr((a_1-a_2)k),
\end{array}\right.\\
\end{eqnarray*}
then we can obtain the probability distribution of the messages. For convenience, denote
\[N\left ( a, b, \alpha , \beta  \right ) =\left | \left \{ k\in \mathbb{F} _{p^{n} }|Tr\left (   ak -k^{p^{r}+  1} \right )=b \ \text{and}\ Tr\left ( \alpha k \right ) =\beta     \right \}  \right |,\]
where $\alpha(\neq 0), a \in \mathbb{F} _{p^{n} }$ and $b , \beta \in \mathbb{F} _{p}$. Let nations be defined as before, based on Lemma \ref{lem11}, we prove Theorem \ref{thm5} by considering the following cases.

If $n\equiv1\,(mod\,2)$, by Lemma \ref{lem11}, we can get
\begin{eqnarray*}
N\left ( a, b, \alpha , \beta  \right )&=&\left\{\begin{array}{ll}
p^{n-2},&  \quad \text{ if $\lambda_{\alpha}=0$ and $\Delta\neq 0$},\\
&\quad \text{or if $\lambda_{\alpha}=0$ and $\Delta=0$ and $\lambda_a-b=0$},\\
p^{n-2}+\frac{1}{p}G_nG_1\eta^\prime(b-\lambda_a), &\quad \text{ if $\lambda_{\alpha}=0$ and $\Delta=0$ and $\lambda_a-b\neq0$},\\
p^{n-2}+\frac{1}{p^2}(p-1)G_nG_1\eta^\prime(-\lambda_{\alpha}), & \quad \text{ if $\lambda_{\alpha}\neq 0$ and $\Delta=0$},\\
p^{n-2}-\frac{1}{p^2}G_nG_1\eta^\prime(-\lambda_{\alpha}), & \quad \text{if $\lambda_{\alpha}\neq 0$ and $\Delta\neq0$}.
\end{array}\right.
\end{eqnarray*}
Based on the discussions above and Lemma \ref{lem10}, the following table is obtained.

\begin{table}[H]
  \centering
  \caption{Ordered-pair message distribution for $n\equiv1\,(mod\,2)$}
  \label{tab:odd-pair-distribution}
  \scriptsize
  \setlength{\tabcolsep}{4pt}
  \renewcommand{\arraystretch}{1.35}
  \resizebox{\textwidth}{!}{%
    \begin{tabular}{l|c|c|c}
      Condition on $(a,b,\alpha,\beta)$ & Conditional uncertainty & Probability & Number of ordered pairs \\
      \hline
      $\lambda_{\alpha}=0$, and either $\Delta\neq0$ or $\Delta=0,\lambda_a-b=0$ & $\log_2 p^{n-2}$ & $\dfrac{p^{n-2}}{(p^n-1)p^{2n}}$ & $(p^2-p+1)(p^{n-1}-1)p^n$\\
      $\lambda_{\alpha}=0,\Delta=0,\lambda_a-b\neq0,\eta^\prime(b-\lambda_a)=1$ & $\log_2(p^{n-2}+p^{\frac{n-1}{2}})$ & $\dfrac{p^{n-2}+p^{\frac{n-1}{2}}}{(p^n-1)p^{2n}}$ & $\dfrac{p-1}{2}(p^{n-1}-1)p^n$\\
      $\lambda_{\alpha}=0,\Delta=0,\lambda_a-b\neq0,\eta^\prime(b-\lambda_a)=-1$ & $\log_2(p^{n-2}-p^{\frac{n-1}{2}})$ & $\dfrac{p^{n-2}-p^{\frac{n-1}{2}}}{(p^n-1)p^{2n}}$ & $\dfrac{p-1}{2}(p^{n-1}-1)p^n$\\
      $\lambda_{\alpha}\neq0,\Delta=0,\eta^\prime(-\lambda_{\alpha})=1$ & $\log_2(p^{n-2}+(p-1)p^{\frac{n-3}{2}})$ & $\dfrac{p^{n-2}+(p-1)p^{\frac{n-3}{2}}}{(p^n-1)p^{2n}}$ & $\dfrac{p-1}{2}p^{n+1}(p^{n-1}+p^{\frac{n-1}{2}})$\\
      $\lambda_{\alpha}\neq0,\Delta=0,\eta^\prime(-\lambda_{\alpha})=-1$ & $\log_2(p^{n-2}-(p-1)p^{\frac{n-3}{2}})$ & $\dfrac{p^{n-2}-(p-1)p^{\frac{n-3}{2}}}{(p^n-1)p^{2n}}$ & $\dfrac{p-1}{2}p^{n+1}(p^{n-1}-p^{\frac{n-1}{2}})$\\
      $\lambda_{\alpha}\neq0,\Delta\neq0,\eta^\prime(-\lambda_{\alpha})=1$ & $\log_2(p^{n-2}-p^{\frac{n-3}{2}})$ & $\dfrac{p^{n-2}-p^{\frac{n-3}{2}}}{(p^n-1)p^{2n}}$ & $\dfrac{(p-1)^2}{2}p^{n+1}(p^{n-1}+p^{\frac{n-1}{2}})$\\
      $\lambda_{\alpha}\neq0,\Delta\neq0,\eta^\prime(-\lambda_{\alpha})=-1$ & $\log_2(p^{n-2}+p^{\frac{n-3}{2}})$ & $\dfrac{p^{n-2}+p^{\frac{n-3}{2}}}{(p^n-1)p^{2n}}$ & $\dfrac{(p-1)^2}{2}p^{n+1}(p^{n-1}-p^{\frac{n-1}{2}})$\\
    \end{tabular}%
  }
\end{table}
Then, let $ Q_1=2^{H(\mathscr{E} |\mathscr{M}^2 )-H(\mathscr{E} |\mathscr{M} )}$,  we have
\[
\begin{split}
&H(\mathscr{E}  |\mathscr{M}^2 )=(n-2)\log_2p+\frac{(p-1)(p^{n-1}-1)}{2p^2(p^n-1)}
 \Biggl[\left(1+p^{-\frac{n-3}{2}}\right)\log_2\left(1+p^{-\frac{n-3}{2}}\right)+\left(1-p^{-\frac{n-3}{2}}\right)
 \log_2\left(1-p^{-\frac{n-3}{2}}\right)\Biggr]\\
&+\frac{(p-1)p^{n-2}}{2(p^n-1)}\Biggl[\left(1+p^{-\frac{n-1}{2}}\right)\left(1+(p-1)p^{-\frac{n-1}{2}}\right)
\cdot\log_2\left(1+(p-1)p^{-\frac{n-1}{2}}\right)\\
&+\left(1-p^{-\frac{n-1}{2}}\right)\left(1-(p-1)p^{-\frac{n-1}{2}}\right)\cdot\log_2\left(1-(p-1)p^{-\frac{n-1}{2}}\right)\Biggr]\\
&+\frac{(p-1)^2p^{n-2}}{2(p^n-1)}\left(1-p^{-(n-1)}\right)\Biggl[\log_2\left(1+p^{-\frac{(n-1)}{2}}\right)
+\log_2\left(1-p^{-\frac{(n-1)}{2}}\right)\Biggr]=(n-2)\log_2p+O\left(p^{-\frac{n-1}{2}}\right),
\end{split}
\]
and
\[Q_1=2^{H(\mathscr{E} |\mathscr{M}^2 )-H(\mathscr{E} |\mathscr{M} )}=\frac{p^{n-2}\cdot 2^{O\left(p^{-\frac{n-1}{2}}\right)}}{p^{n-1}\cdot 2^{O\left(p^{-\frac{n-1}{2}}\right)}}\]

By Theorem \ref{thm2}, we have $P_S=\frac{p^{n-2}+(p-1)p^{\frac{n-3}{2}}}{p^{n-1}}$, it follows that
\[\lim_{n \to \infty} \frac{Q_1}{P_{S}} =1.\]

If $n\equiv0\,(mod\,2)$ and $\frac{n}{v}\equiv1 \,(mod\,2)$ or $\frac{n}{v} \equiv2\,(mod\,4)$, a similar calculation to that in the first case yields the result
\[\lim_{n \to \infty} \frac{Q_1}{P_{S}} =1,\]
the computational details are omitted here to save space.

If $\frac{n}{v} \equiv0\,(mod\,4)$. Let $L(X)=X^{p^{2r}}+X$ and define
\[W=\left\{c\in\mathbb{F}_{p^n}:L(X)=-c^{p^r} \text{ is solvable in }\mathbb{F}_{p^n}\right\}.\]
Since $L$ and $c\mapsto-c^{p^r}$ are $\mathbb{F}_p$-linear maps, $W$ is an $\mathbb{F}_p$-linear subspace. Moreover, by Lemmas \ref{lem4} and \ref{lem6}, $|\ker L|=p^{2v}$ and $|W|=p^{n-2v}$.

For $c\in W$, choose a solution $\gamma_c$ of
$L(X)=-c^{p^r}$, and put $\lambda_c=Tr(\gamma_c^{p^r+1})$. $\lambda_c$ is independent of the choice of $\gamma_c$. Indeed, if
$\gamma_c+z$ is another solution, where $z\in\ker L$ with $z^{p^{2r}}=-z$, we will show $Tr((\gamma_c+z)^{p^r+1})=Tr(\gamma_c^{p^r+1})$.
Since
\[Tr((\gamma_c+z)^{p^r+1})=Tr(\gamma_c^{p^r+1}+\gamma_c^{p^r}z+\gamma_cz^{p^r}+z^{p^r+1}),\]
and raising $z^{p^{2r}}+z=0$ to the $p^{n-r}$-th power gives $z^{p^{n-r}}+z^{p^r}=0$,
by the invariance of the trace under the Frobenius map, we have
\[Tr(\gamma_c^{p^r}z+\gamma_cz^{p^r})=Tr\bigl(\gamma_c(z^{p^{n-r}}+z^{p^r})\bigr)=0.\]
Moreover
\[Tr(z^{p^r+1})=Tr\bigl((z^{p^r+1})^{p^r}\bigr)=-Tr(z^{p^r+1}),\]
so $Tr(z^{p^r+1})=0$ because $p$ is odd. Thus, we have $Tr((\gamma_c+z)^{p^r+1})=Tr(\gamma_c^{p^r+1})$. Therefore, for a given $c\in W$, by Lemma \ref{lem4} and Lemma \ref{lem6}, the value of $Tr(x^{p^r+1})$ is the same for all $x$ satisfying the equation $L(x)=-c^{p^r}$.

We count the elements of $W$ according as $\lambda_c$ is zero or nonzero. By Lemma \ref{lem10}, with $a=b=0$,
\[\left|\left\{\gamma\in\mathbb{F}_{p^n}:Tr(\gamma^{p^r+1})=0\right\}\right|=p^{n-1}-(p-1)p^{\frac{n}{2}+v-1}.\]
It follows, after division by $p^{2v}$ and exclusion of $c=0$, that
the number of nonzero $\alpha\in W$ satisfying $\lambda_\alpha=0$ is $p^{n-2v-1}-(p-1)p^{\frac{n}{2}-v-1}-1$, and the number satisfying $\lambda_\alpha\neq0$ is $p^{n-2v}-p^{n-2v-1}+(p-1)p^{\frac{n}{2}-v-1}$.

We now determine the contribution of every solvability case.

{\bf Case 1.} $a\in W$ and $\alpha\in W\setminus\{0\}$. There are $p^{n-2v}$ choices for $a$.
If $\lambda_\alpha=0$, then for fixed $a$ and $\alpha$, Lemma \ref{lem11} gives
\[ N(a,b,\alpha,\beta)=
 \begin{cases}
 p^{n-2},&\Delta\neq 0,\\
 p^{n-2}-(p-1)p^{\frac{n}{2}+v-1},
   &\Delta=0,\ \lambda_a-b=0,\\
 p^{n-2}+p^{\frac{n}{2}+v-1},
   &\Delta=0,\ \lambda_a-b\neq 0.
 \end{cases}\]
The three alternatives contain, respectively, $p(p-1)$, $1$, and $p-1$ choices of $(b,\beta)$.
Thus the contributions of this subcase to the preceding three fiber
sizes, in order, are
\begin{align*}
&p^{n-2v}p(p-1)\left(p^{n-2v-1}-(p-1)p^{\frac n2-v-1}-1\right),\\
&p^{n-2v}\left(p^{n-2v-1}-(p-1)p^{\frac n2-v-1}-1\right),  \\
&p^{n-2v}(p-1)\left(p^{n-2v-1}-(p-1)p^{\frac n2-v-1}-1\right).
\end{align*}

If $\lambda_\alpha\neq0$, then for fixed $a$ and $\alpha$ Lemma \ref{lem11} gives
\[ N(a,b,\alpha,\beta)=
 \begin{cases}
 p^{n-2},&\Delta=0,\\
 p^{n-2}+p^{\frac{n}{2}+v-1},
   &\Delta\neq0\ \eta^\prime(\Delta)=-1,\\
 p^{n-2}-p^{\frac{n}{2}+v-1},
   &\Delta\neq0\ \eta^\prime(\Delta)=1.\\
 \end{cases}\]
The map $(b,\beta)\longmapsto\Delta=4\lambda_\alpha(\lambda_a-b)-(\lambda_{\alpha a}-\beta)^2$ has exactly $p$ preimages for each $\Delta\in\mathbb{F}_p$. Indeed, after $\beta$ and $\Delta$ are prescribed, $b$ is uniquely determined by $b=\lambda_a-\frac{\Delta+(\lambda_{\alpha a}-\beta)^2}{4\lambda_\alpha}$. Thus $\Delta=0$ occurs for $p$ pairs $(b,\beta)$, whereas $\eta^\prime(\Delta)=1$ and $\eta^\prime(\Delta)=-1$ each occur for $p(p-1)/2$ pairs. Thus the contributions of this subcase to the preceding three fiber sizes are
\begin{align*}
&p^{n-2v}p\left(p^{n-2v}-p^{n-2v-1}+(p-1)p^{\frac n2-v-1}\right),\\
&\frac{p^{n-2v}p(p-1)}2\left(p^{n-2v}-p^{n-2v-1}+(p-1)p^{\frac n2-v-1}\right),\\
&\frac{p^{n-2v}p(p-1)}2 \left(p^{n-2v}-p^{n-2v-1}+(p-1)p^{\frac n2-v-1}\right).
\end{align*}

{\bf Case 2.} $a\in W$ and $\alpha\notin W$. There are $p^{n-2v}(p^n-p^{n-2v})$ choices of $(a,\alpha)$. For fixed $a$ and $\alpha$, Lemma \ref{lem11} gives
\[ N(a,b,\alpha,\beta)=
 \begin{cases}
 p^{n-2}+p^{\frac{n}{2}+v-2}, &\lambda_a-b\neq 0,\\
 p^{n-2}-(p-1)p^{\frac{n}{2}+v-2},
   & \lambda_a-b=0.
 \end{cases}\]
Thus the contributions of this case to these two fiber sizes are,
respectively,
\[p^{n-2v}(p^n-p^{n-2v})p,\quad p^{n-2v}(p^n-p^{n-2v})p(p-1). \]

{\bf Case 3.} $a\notin W$ and $\alpha\in W\setminus\{0\}$. There are $(p^n-p^{n-2v})(p^{n-2v}-1)$ choices of $(a,\alpha)$. For fixed $a$ and $\alpha$, Lemma \ref{lem11} gives $N(a,b,\alpha,\beta)=p^{n-2}$. Hence the contribution of this case to the fiber size is
$(p^n-p^{n-2v})(p^{n-2v}-1)p^2$.

{\bf Case 4.} $a\notin W$ and $\alpha\notin W$, Lemma \ref{lem11} gives
\[ N(a,b,\alpha,\beta)=
 \begin{cases}
 p^{n-2},&\text{no $t_0$ exists},\\
 p^{n-2}-(p-1)p^{\frac{n}{2}+v-2},
   &\lambda_{t_0}-b-t_0\beta=0,\\
 p^{n-2}+p^{\frac{n}{2}+v-2},
   & \lambda_{t_0}-b-t_0\beta\neq 0.
 \end{cases}\]
In the quotient space $\mathbb{F}_{p^n}/W$, the condition $a+t\alpha\in W$ is equivalent to $a+W=-t(\alpha+W)$. For fixed $\alpha\notin W$, as $t$ ranges over $\mathbb{F}_p^*$, the right-hand side gives $p-1$ distinct nonzero cosets, and every coset contains $p^{n-2v}$ elements. Therefore, the number of pairs $(a,\alpha)$ for which there exists a unique $t_0\in\mathbb{F}_p^*$ satisfying $a+t_0\alpha\in W$ is
$(p^n-p^{n-2v})(p-1)p^{n-2v}$. Again fix $\alpha\notin W$, there are $p^n-p^{n-2v}$ choices of $a\notin W$, and the number of choices of $a$ for which no such $t_0$ exists is $p^n-p^{n-2v}-(p-1)p^{n-2v}=p^n-p^{n-2v+1}$, where $(p-1)p^{n-2v}$ is the number of $a$ such that $a+t_0\alpha\in W$. Thus the contributions of this subcase to the preceding three fiber sizes are
\[(p^n-p^{n-2v})(p^n-p^{n-2v+1})p^2,\ (p^n-p^{n-2v})(p-1)p^{n-2v}p,\ (p^n-p^{n-2v})(p-1)p^{n-2v}p(p-1).\]

Let $\mathcal C(t)$ be the number of parameters $(a,b,\alpha,\beta)$
for which $N(a,b,\alpha,\beta)=t$. Equivalently, $\mathcal C(t)$
is the number of ordered pairs of distinct messages having fiber size
$t$. Adding the contributions of the preceding cases gives the following table.
\begin{table}[H]
  \centering
  \caption{Ordered-pair message distribution for $\frac{n}{v} \equiv0\,(mod\,4)$}
  \label{tab:even-fiber-distribution}
  \scriptsize
  \setlength{\tabcolsep}{4pt}
  \renewcommand{\arraystretch}{1.35}
  \resizebox{\textwidth}{!}{%
\begin{tabular}{l|c|c|c}
Fiber size $t$ & Conditional uncertainty
& Probability & Number of ordered pairs $\mathcal C(t)$\\
\hline
$p^{n-2}-(p-1)p^{\frac n2+v-1}$
& $\log_2\left(p^{n-2}-(p-1)p^{\frac n2+v-1}\right)$
& $\dfrac{p^{n-2}-(p-1)p^{\frac n2+v-1}}{(p^n-1)p^{2n}}$
& $p^{n-2v}\left(p^{n-2v-1}-(p-1)p^{\frac n2-v-1}-1\right)$\\
$p^{n-2}+p^{\frac n2+v-1}$
& $\log_2\left(p^{n-2}+p^{\frac n2+v-1}\right)$
& $\dfrac{p^{n-2}+p^{\frac n2+v-1}}{(p^n-1)p^{2n}}$
& $\begin{aligned}
p^{n-2v}(p-1)&\left(p^{n-2v-1}-(p-1)p^{\frac n2-v-1}-1\right)\\
&+\dfrac{p^{n-2v}p(p-1)}2
\left(p^{n-2v}-p^{n-2v-1}+(p-1)p^{\frac n2-v-1}\right)
\end{aligned}$\\
$p^{n-2}-p^{\frac n2+v-1}$
& $\log_2\left(p^{n-2}-p^{\frac n2+v-1}\right)$
& $\dfrac{p^{n-2}-p^{\frac n2+v-1}}{(p^n-1)p^{2n}}$
& $\dfrac{p^{n-2v}p(p-1)}2
\left(p^{n-2v}-p^{n-2v-1}+(p-1)p^{\frac n2-v-1}\right)$\\
$p^{n-2}-(p-1)p^{\frac n2+v-2}$
& $\log_2\left(p^{n-2}-(p-1)p^{\frac n2+v-2}\right)$
& $\dfrac{p^{n-2}-(p-1)p^{\frac n2+v-2}}{(p^n-1)p^{2n}}$
& $p^2p^{n-2v}(p^n-p^{n-2v})$\\
$p^{n-2}+p^{\frac n2+v-2}$
& $\log_2\left(p^{n-2}+p^{\frac n2+v-2}\right)$
& $\dfrac{p^{n-2}+p^{\frac n2+v-2}}{(p^n-1)p^{2n}}$
& $p^2(p-1)p^{n-2v}(p^n-p^{n-2v})$\\
$p^{n-2}$
& $\log_2 p^{n-2}$
& $\dfrac{p^{n-2}}{(p^n-1)p^{2n}}$
& $\begin{aligned}
&p^{n-2v}p(p-1)
\left(p^{n-2v-1}-(p-1)p^{\frac n2-v-1}-1\right)\\
&+(p^n-p^{n-2v})(p^{n-2v}-1)p^2\\
&+p^{n-2v}p
\left(p^{n-2v}-p^{n-2v-1}+(p-1)p^{\frac n2-v-1}\right)\\
&+(p^n-p^{n-2v})(p^n-p^{n-2v+1})p^2
\end{aligned}$\\
\end{tabular}%
  }
\end{table}

Let $\Phi(x)=x\log_2x$ for $x>0$ and $\Phi(0)=0$, then we have
\[H(\mathscr{E}|\mathscr{M}^2) =\frac{1}{(p^n-1)p^{2n}}\sum_{t:\,\mathcal C(t)>0}\mathcal C(t)\Phi(t).\]
Since $\frac{n}{v} \equiv0\,(mod\,4)$, $v\leq n/4$, for fixed $p$, every nonzero fiber size having positive multiplicity is
\[p^{n-2}\left(1+O\left(p^{-\frac{n}{2}+v}\right)\right),\]
and it is easy to verify that
\[\sum_t\mathcal C(t)=p^n(p^n-1)p^2,\qquad\sum_t\mathcal C(t)t=p^{2n}(p^n-1).\]
Thus, one can obtain directly
\[H(\mathscr{E}|\mathscr{M}^2)=(n-2)\log_2p+O\left(p^{-\frac{n}{2}+v}\right).\]
From the proof of Theorem \ref{thm4}, we have
\[H(\mathscr{E}|\mathscr{M})=(n-1)\log_2p+O\left(p^{-\frac{n}{2}+v}\right).\]
Consequently,
\[Q_1=2^{H(\mathscr{E}|\mathscr{M}^2)-H(\mathscr{E}|\mathscr{M})}=\frac{p^{n-2}\cdot2^{O\left(p^{-\frac{n}{2}+v}\right)}}
{p^{n-1}2^{O\left(p^{-\frac{n}{2}+v}\right)}}.\]
On the other hand, by Theorem \ref{thm2}, $P_S=\frac{p^{n-2}}{p^{n-1}-(p-1)p^{\frac{n}{2}+v-1}}$.
Therefore, we obtain that
\[\lim_{n\to\infty}\frac{Q_1}{P_S}=1.\]
This completes the proof of Theorem \ref{thm5}.
\end{proof}

We now compare our results with those in \cite{ref2} \cite{ref3} \cite{ref14} \cite{ref15} \cite{ref16}, the comparison is summarized in the following table.

\begin{table}[htbp]
\centering
\caption{Compare the performance of the authentication codes in the references}
\label{tab:example}
\resizebox{\linewidth}{!}{%
\begin{tabular}{|c|c|c|c|c|}
\hline
Authentication code & $P_I$  & $P_S$  &The information-theoretic bound &The combinatorial bound \\\hline
Construction 1 in \cite{ref2} & exact value & exact value & optimal($P_S$) &  optimal($P_I$) \\\hline
Construction 2 in \cite{ref2} & exact value & exact value & optimal($P_I$) & asymptotically optimal($P_S$) \\\hline
Construction 3 in \cite{ref2} & exact value & upper bound & optimal($P_S$) & optimal($P_I$) \\\hline
Construction 1 in \cite{ref3} & exact value & upper bound & not computed & asymptotically optimal($P_I$ and $P_S$) \\
\hline
Construction 2 in \cite{ref3} & upper bound & upper bound & not computed & asymptotically optimal($P_I$ and $P_S$) \\
\hline
Construction 3 in \cite{ref3} & exact value & upper bound&  asymptotically optimal($P_I$)& asymptotically optimal($P_I$ and $P_S$) \\
\hline
Constructions in \cite{ref14} & exact value or upper bound & upper bound & not computed &  optimal($P_I$) asymptotically optimal($P_S$)\\
\hline
Constructions in \cite{ref15} & upper bound & upper bound & not computed & not computed\\
\hline
Constructions 1-7 in \cite{ref16} & exact value or upper bound & upper bound & not computed & optimal($P_I$) or asymptotically optimal($P_S$)\\\hline
Construction in this paper & exact value & exact value & asymptotically optimal($P_I$ and $P_S$)  & asymptotically optimal($P_I$ and $P_S$)  \\
\hline
\end{tabular}
}
\end{table}

As can be seen from the table, we obtain the exact values of the authentication code $P_I$ and $P_S$, and prove that both $P_I$ and $P_S$ asymptotically attain the combinatorial bound and the information-theoretic bound. In particular, only this paper proves that $P_I$and $P_S$ are asymptotically optimal with respect to both the information-theoretic bound and the combinatorial bound.

\subsection{Concluding remarks}
This paper constructs a class of authentication codes with secrecy using algebraic structures over finite fields. The codes have simple encoding rules. The security of the authentication codes is analyzed against impersonation attacks and substitution attacks. Using Weil sums and exponential sums, the maximum success probabilities $P_{I} $ and $P_{S} $ are explicitly derived. The codes are shown to be asymptotically optimal, as both $P_{I} $ and $P_{S} $ asymptotically achieve the information-theoretic and combinatorial lower bounds.

\end{document}